\documentclass[journal]{IEEEtran}
\usepackage{array,url}
\usepackage[geometry]{ifsym}
\usepackage{amssymb,amsfonts,amsmath,amsthm,amscd,mathrsfs, mathtools}
\usepackage[textsize=footnotesize]{todonotes}
\usepackage{graphicx, float, color, comment, enumerate, verbatim, fixmath, bbm, latexsym, scrextend}
\usepackage[nice]{nicefrac}
\usepackage[hidelinks]{hyperref}

\makeatletter
\newcommand\footnoteref[1]{\protected@xdef\@thefnmark{\ref{#1}}\@footnotemark}
\makeatother

\usepackage{cleveref}
\crefname{apdx}{Appendix}{Appendices}

\usepackage{pgf,tikz}
\usepackage{pgfplots}

\newcommand{\bbC}{\mathbb{C}}
\newcommand{\rmd}{\mathrm{d}}
\newcommand{\bbE}{\mathbb{E}}\newcommand{\rme}{\mathrm{e}}

\newcommand{\bbI}{\mathbb{I}}\newcommand{\rmI}{\mathrm{I}}

\newcommand{\bbN}{\mathbb{N}}\newcommand{\rmN}{\mathrm{N}}\newcommand{\rmn}{\mathrm{n}}
\newcommand{\rmO}{\mathrm{O}}

\newcommand{\bbR}{\mathbb{R}}

\newcommand{\sfA}{\mathsf{A}}
\newcommand{\sfB}{\mathsf{B}}

\newcommand{\sfD}{\mathsf{D}}

\newcommand{\bfI}{\mathbf{I}}\newcommand{\sfI}{\mathsf{I}}

\newcommand{\sfM}{\mathsf{M}}

\newcommand{\sfP}{\mathsf{P}}

\newcommand{\cD}{\mathcal{D}}

\newcommand{\cI}{\mathcal{I}}

\newcommand{\cN}{\mathcal{N}}\newcommand{\scrN}{\mathscr{N}}
\newcommand{\cO}{\mathcal{O}}

\newcommand{\scrS}{\mathscr{S}}

\newcommand{\cX}{\mathcal{X}}

 \newcommand{\mfra}{\mathfrak{a}}

\newcommand{\mfrF}{\mathfrak{F}}

\newtheoremstyle{mystyle}
{}
{}
{\itshape}
{}
{\bfseries}
{}
{.5em}
{}

\newtheoremstyle{remark}
{}
{}
{}
{}
{\itshape}
{}
{.5em}
{}

\makeatletter
\def\thmhead@plain#1#2#3{%
  \thmname{#1}\thmnumber{\@ifnotempty{#1}{ }\@upn{#2.}}%
  \thmnote{ \small\textsf{\the\thm@notefont\textit{#3}.}}}
\let\thmhead\thmhead@plain
\makeatother

\theoremstyle{mystyle}
\newtheorem{theorem}{Theorem}
\theoremstyle{mystyle}
\newtheorem{lemma}{Lemma}
\theoremstyle{mystyle}
\theoremstyle{mystyle}
\newtheorem{corollary}{Corollary}
\theoremstyle{mystyle}
\newtheorem{definition}{Definition}
\theoremstyle{mystyle}
\theoremstyle{mystyle}
\theoremstyle{mystyle}
\theoremstyle{mystyle}
\theoremstyle{mystyle}
\theoremstyle{mystyle}

\theoremstyle{remark}
\newtheorem{rem}{Remark}

\usepackage{subcaption}
\newcommand\independent{\protect\mathpalette{\protect\independent}{\perp}}
\def\independent#1#2{\mathrel{\rlap{$#1#2$}\mkern2mu{#1#2}}}

\def\squarebox#1{\hbox to #1{\hfill\vbox to #1{\vfill}}}

\newcommand{\opTheta}{\operatorname{\Theta}}

\newcommand{\del}{\partial}

\newcommand{\opGamma}{\operatorname{\Gamma}}

\makeatletter
\newcommand{\bigger}{\bBigg@{3}}
\newcommand{\vast}{\bBigg@{4}}
\newcommand{\Vast}{\bBigg@{5}}
\newcommand{\Gigantic}{\bBigg@{8}}

\makeatother

\newcommand{\bsX}{\boldsymbol{X}}
\newcommand{\bsY}{\boldsymbol{Y}}
\newcommand{\bsZ}{\boldsymbol{Z}}

\allowdisplaybreaks 

\renewcommand{\qedsymbol}{$\blacksquare$}

\newcommand{\supp}{{\mathsf{supp}}}

\title{The Capacity Achieving Distribution for the Amplitude Constrained Additive Gaussian Channel:  An Upper Bound on the Number of Mass Points}  
\author{Alex Dytso,~\IEEEmembership{Member, IEEE} Semih Yagli,~\IEEEmembership{Student Member, IEEE} H. Vincent Poor,~\IEEEmembership{Fellow, IEEE}, and Shlomo Shamai (Shitz),~\IEEEmembership{Fellow, IEEE}%
\thanks{The work of A. Dytso, S. Yagli and H. V. Poor was supported by the U. S. National Science Foundation under Grant CCF-1908308 and in part by the United States-Israel Binational Science Foundation, under Grant BSF-2018710. The work of S. Shamai  was supported by the United States-Israel Binational Science Foundation,
under Grant BSF-2018710, and in part by the European Union's Horizon
2020 Research And Innovation Programme, grant agreement no. 694630.  
 }
\thanks{A. Dytso, S. Yagli, and H. V. Poor are with the Department of Electrical Engineering, Princeton University, Princeton, NJ 08544, USA   (e-mail: \{adytso,  syagli, poor\}@princeton.edu).}
\thanks{S.\,Shamai (Shitz) is with the Department of Electrical Engineering, Technion -- Israel Institute of Technology, Haifa 3200003, Israel (e-mail: sshlomo@ee.technion.ac.il).
}
}

\begin{document}
\maketitle

\begin{abstract}

This paper studies an $n$-dimensional additive  Gaussian noise channel with  a peak-power-constrained  input. 
 It is well known that, in this case,  when $n=1$ the capacity-achieving input distribution is discrete with finitely many mass points, and when $n>1$  the capacity-achieving input distribution is supported on finitely many concentric shells. 
 However, due to the previous proof technique, not even a bound on the exact number of  mass points/shells was available.
  This paper provides an alternative proof of the finiteness of the number mass points/shells of the capacity-achieving input distribution while producing the first firm bounds on the number of mass points and shells, paving an alternative way for approaching many such problems.  
 
 The first main result of this paper is an order tight implicit bound which shows that the number of mass points in the capacity-achieving input distribution is within a factor of two from the number of zeros of the downward shifted capacity-achieving output probability density function. Next, this implicit bound is utilized to provide a first firm upper on the support size of optimal input distribution, an $O(\sfA^2)$ upper bound where $\sfA$ denotes the constraint on the input amplitude. The second main result of this paper generalizes the first one to the case when $n>1$, showing that, for each and every dimension $n\ge 1$, the number of shells that the optimal input distribution contains is $O(\sfA^2)$. Finally, the third main result of this paper reconsiders the case $n=1$ with an additional average power constraint, demonstrating a similar $O(\sfA^2)$ bound.

\end{abstract}

\section{Introduction}
We consider  an additive noise channel for which the input-output relationship is given by 
\begin{align}
\bsY=\bsX+\bsZ \text{,} \label{eq:channel}   
\end{align}
where the input $\bsX \in \mathbb{R}^n$ is independent of the standard Gaussian noise  $\bsZ \in \mathbb{R}^n $.  We are interested in finding the capacity of the channel  in \eqref{eq:channel} subject to the constraint that $\bsX \in \mathcal{B}_{0}(\sfA) $ where $\mathcal{B}_{0}(\sfA) $ is an $n$-ball centered at zero with radius $\sfA$ (i.e., amplitude or peak-power constrained input), that is
\begin{align}
C_n(\sfA)=\max_{\bsX \colon  \bsX \in  \mathcal{B}_{0}(\sfA) } I(\bsX;\bsY) \text{.} \label{eq:Capacity} 
\end{align}

In his seminal paper \cite{smith1971information} (see also \cite{smith1969Thesis}), for the case of $n=1$, Smith has shown that an optimizing distribution in \eqref{eq:Capacity} is unique, symmetric around the origin, and  perhaps surprisingly, discrete with finitely many mass points. Using tools such as the Identity Theorem from complex analysis, Smith has proven that the cardinality of the support set of the optimal input distribution cannot be infinite, and, thus, must be  finite.  
Employing this proof by contradiction, Shamai and Bar-David \cite{ShamQuadrat} have extended the method of Smith to $n=2$, and showed that, in this setting, the maximizing input random variable is given by
\begin{align}
\bsX^\star=R^\star \cdot \boldsymbol{U}^\star \label{eq:maximizingInput}
\end{align}
where  the magnitude $R^\star$ is discrete with finitely many points  and the random unit vector $ \boldsymbol{U}^\star$, which is independent of $R^\star$, has a uniform phase on $[0,2\pi)$. In other words, the  support is given by finitely many concentric shells, e.g., Fig.~\ref{fig:InputSphere}. 
  As a matter of fact, this phenomena that the optimal input distribution lies on finitely many concentric spheres remains true for any $n\ge 2$, cf. \cite{chan2005MIMObounded,rassouli2016capacity} and \cite{CISS2018}. 
  
\begin{figure}
\center
%
%
\pgfplotsset{every axis plot/.append style={very thick}}
\begin{tikzpicture}

\begin{axis}[%
width=6.953cm,
height=6.953cm,
at={(1.166in,0.84in)},
scale only axis,
xmin=-3,
xmax=3,
xlabel style={font=\color{white!15!black}},
xlabel={$x_1$},
ymin=-3,
ymax=3,
ylabel style={font=\color{white!15!black}},
ylabel={$x_2$},
axis background/.style={fill=white},
xmajorgrids,
ymajorgrids
]
\addplot [color=black, line width=2.0pt, forget plot]
  table[row sep=crcr]{%
1.7	0\\
1.69664543832806	0.106743883199833\\
1.68659499223461	0.213066497059317\\
1.66988832623877	0.318548234795732\\
1.64659137391867	0.422772808180253\\
1.61679607770176	0.525328890437411\\
1.58062002601003	0.625811739563952\\
1.53820598919223	0.723824795660624\\
1.48972135607457	0.818981245972916\\
1.43535747335343	0.910905551464294\\
1.37532889043741	0.999234928897204\\
1.30987251271884	1.08362078257277\\
1.2392466666164	1.16373008007877\\
1.16373008007877	1.2392466666164\\
1.08362078257277	1.30987251271884\\
0.999234928897204	1.37532889043741\\
0.910905551464294	1.43535747335343\\
0.818981245972916	1.48972135607457\\
0.723824795660624	1.53820598919223\\
0.625811739563953	1.58062002601003\\
0.525328890437411	1.61679607770176\\
0.422772808180253	1.64659137391867\\
0.318548234795732	1.66988832623877\\
0.213066497059317	1.68659499223461\\
0.106743883199833	1.69664543832806\\
1.04094977927525e-16	1.7\\
-0.106743883199833	1.69664543832806\\
-0.213066497059317	1.68659499223461\\
-0.318548234795732	1.66988832623877\\
-0.422772808180253	1.64659137391867\\
-0.52532889043741	1.61679607770176\\
-0.625811739563953	1.58062002601003\\
-0.723824795660624	1.53820598919223\\
-0.818981245972916	1.48972135607457\\
-0.910905551464295	1.43535747335343\\
-0.999234928897204	1.37532889043741\\
-1.08362078257277	1.30987251271884\\
-1.16373008007877	1.2392466666164\\
-1.2392466666164	1.16373008007877\\
-1.30987251271884	1.08362078257277\\
-1.37532889043741	0.999234928897204\\
-1.43535747335343	0.910905551464294\\
-1.48972135607457	0.818981245972917\\
-1.53820598919223	0.723824795660624\\
-1.58062002601003	0.625811739563953\\
-1.61679607770176	0.525328890437411\\
-1.64659137391867	0.422772808180253\\
-1.66988832623877	0.318548234795732\\
-1.68659499223461	0.213066497059318\\
-1.69664543832806	0.106743883199833\\
-1.7	2.0818995585505e-16\\
-1.69664543832806	-0.106743883199833\\
-1.68659499223461	-0.213066497059317\\
-1.66988832623877	-0.318548234795732\\
-1.64659137391867	-0.422772808180254\\
-1.61679607770176	-0.525328890437411\\
-1.58062002601003	-0.625811739563953\\
-1.53820598919223	-0.723824795660624\\
-1.48972135607457	-0.818981245972916\\
-1.43535747335343	-0.910905551464294\\
-1.37532889043741	-0.999234928897204\\
-1.30987251271884	-1.08362078257277\\
-1.2392466666164	-1.16373008007877\\
-1.16373008007877	-1.2392466666164\\
-1.08362078257277	-1.30987251271884\\
-0.999234928897206	-1.37532889043741\\
-0.910905551464295	-1.43535747335343\\
-0.818981245972917	-1.48972135607457\\
-0.723824795660624	-1.53820598919223\\
-0.625811739563954	-1.58062002601003\\
-0.525328890437411	-1.61679607770176\\
-0.422772808180254	-1.64659137391867\\
-0.318548234795732	-1.66988832623877\\
-0.213066497059318	-1.68659499223461\\
-0.106743883199832	-1.69664543832806\\
-3.12284933782575e-16	-1.7\\
0.106743883199832	-1.69664543832806\\
0.213066497059317	-1.68659499223461\\
0.318548234795731	-1.66988832623877\\
0.422772808180253	-1.64659137391867\\
0.52532889043741	-1.61679607770176\\
0.625811739563953	-1.58062002601003\\
0.723824795660623	-1.53820598919223\\
0.818981245972915	-1.48972135607457\\
0.910905551464293	-1.43535747335343\\
0.999234928897204	-1.37532889043741\\
1.08362078257277	-1.30987251271884\\
1.16373008007877	-1.2392466666164\\
1.2392466666164	-1.16373008007877\\
1.30987251271884	-1.08362078257277\\
1.37532889043741	-0.999234928897205\\
1.43535747335343	-0.910905551464294\\
1.48972135607457	-0.818981245972916\\
1.53820598919223	-0.723824795660624\\
1.58062002601003	-0.625811739563954\\
1.61679607770176	-0.525328890437411\\
1.64659137391867	-0.422772808180254\\
1.66988832623877	-0.318548234795732\\
1.68659499223461	-0.213066497059318\\
1.69664543832806	-0.106743883199833\\
1.7	-4.163799117101e-16\\
};
\addplot [color=black, line width=2.0pt, forget plot]
  table[row sep=crcr]{%
1	0\\
0.998026728428272	0.0627905195293134\\
0.992114701314478	0.125333233564304\\
0.982287250728689	0.187381314585725\\
0.968583161128631	0.248689887164855\\
0.951056516295154	0.309016994374947\\
0.929776485888251	0.368124552684678\\
0.904827052466019	0.425779291565073\\
0.876306680043864	0.481753674101715\\
0.844327925502015	0.535826794978997\\
0.809016994374947	0.587785252292473\\
0.770513242775789	0.63742398974869\\
0.728968627421412	0.684547105928689\\
0.684547105928689	0.728968627421412\\
0.63742398974869	0.770513242775789\\
0.587785252292473	0.809016994374947\\
0.535826794978997	0.844327925502015\\
0.481753674101715	0.876306680043864\\
0.425779291565073	0.90482705246602\\
0.368124552684678	0.929776485888251\\
0.309016994374947	0.951056516295154\\
0.248689887164855	0.968583161128631\\
0.187381314585725	0.982287250728689\\
0.125333233564304	0.992114701314478\\
0.0627905195293135	0.998026728428272\\
6.12323399573677e-17	1\\
-0.0627905195293134	0.998026728428272\\
-0.125333233564304	0.992114701314478\\
-0.187381314585725	0.982287250728689\\
-0.248689887164855	0.968583161128631\\
-0.309016994374947	0.951056516295154\\
-0.368124552684678	0.929776485888251\\
-0.425779291565073	0.904827052466019\\
-0.481753674101715	0.876306680043863\\
-0.535826794978997	0.844327925502015\\
-0.587785252292473	0.809016994374947\\
-0.63742398974869	0.770513242775789\\
-0.684547105928689	0.728968627421411\\
-0.728968627421411	0.684547105928689\\
-0.770513242775789	0.63742398974869\\
-0.809016994374947	0.587785252292473\\
-0.844327925502015	0.535826794978997\\
-0.876306680043863	0.481753674101716\\
-0.904827052466019	0.425779291565073\\
-0.929776485888251	0.368124552684678\\
-0.951056516295154	0.309016994374948\\
-0.968583161128631	0.248689887164855\\
-0.982287250728689	0.187381314585725\\
-0.992114701314478	0.125333233564305\\
-0.998026728428272	0.0627905195293136\\
-1	1.22464679914735e-16\\
-0.998026728428272	-0.0627905195293133\\
-0.992114701314478	-0.125333233564304\\
-0.982287250728689	-0.187381314585725\\
-0.968583161128631	-0.248689887164855\\
-0.951056516295154	-0.309016994374948\\
-0.929776485888251	-0.368124552684678\\
-0.904827052466019	-0.425779291565073\\
-0.876306680043864	-0.481753674101715\\
-0.844327925502015	-0.535826794978996\\
-0.809016994374947	-0.587785252292473\\
-0.770513242775789	-0.63742398974869\\
-0.728968627421412	-0.684547105928689\\
-0.684547105928689	-0.728968627421411\\
-0.63742398974869	-0.770513242775789\\
-0.587785252292474	-0.809016994374947\\
-0.535826794978997	-0.844327925502015\\
-0.481753674101716	-0.876306680043863\\
-0.425779291565073	-0.904827052466019\\
-0.368124552684679	-0.929776485888251\\
-0.309016994374948	-0.951056516295154\\
-0.248689887164855	-0.968583161128631\\
-0.187381314585725	-0.982287250728689\\
-0.125333233564305	-0.992114701314478\\
-0.0627905195293132	-0.998026728428272\\
-1.83697019872103e-16	-1\\
0.0627905195293128	-0.998026728428272\\
0.125333233564304	-0.992114701314478\\
0.187381314585724	-0.982287250728689\\
0.248689887164855	-0.968583161128631\\
0.309016994374947	-0.951056516295154\\
0.368124552684678	-0.929776485888251\\
0.425779291565073	-0.90482705246602\\
0.481753674101715	-0.876306680043864\\
0.535826794978996	-0.844327925502016\\
0.587785252292473	-0.809016994374948\\
0.637423989748689	-0.77051324277579\\
0.684547105928689	-0.728968627421412\\
0.728968627421411	-0.684547105928689\\
0.770513242775789	-0.63742398974869\\
0.809016994374947	-0.587785252292473\\
0.844327925502015	-0.535826794978996\\
0.876306680043864	-0.481753674101715\\
0.904827052466019	-0.425779291565073\\
0.929776485888251	-0.368124552684679\\
0.951056516295154	-0.309016994374948\\
0.968583161128631	-0.248689887164855\\
0.982287250728689	-0.187381314585725\\
0.992114701314478	-0.125333233564305\\
0.998026728428272	-0.0627905195293133\\
1	-2.44929359829471e-16\\
};
\addplot [color=black, line width=2.0pt, forget plot]
  table[row sep=crcr]{%
2.1	0\\
2.09585612969937	0.131860091011558\\
2.0834408727604	0.263199790485039\\
2.06280322653025	0.393500760630022\\
2.03402463837013	0.522248763046195\\
1.99721868421982	0.64893568818739\\
1.95253062036533	0.773061560637824\\
1.90013681017864	0.894136512286653\\
1.84024402809211	1.0116827156136\\
1.77308864355423	1.12523626945589\\
1.69893568818739	1.23434902981419\\
1.61807780982916	1.33859037847225\\
1.53083411758496	1.43754892245025\\
1.43754892245025	1.53083411758496\\
1.33859037847225	1.61807780982916\\
1.23434902981419	1.69893568818739\\
1.12523626945589	1.77308864355423\\
1.0116827156136	1.84024402809211\\
0.894136512286653	1.90013681017864\\
0.773061560637824	1.95253062036533\\
0.64893568818739	1.99721868421982\\
0.522248763046195	2.03402463837013\\
0.393500760630022	2.06280322653025\\
0.263199790485039	2.0834408727604\\
0.131860091011558	2.09585612969937\\
1.28587913910472e-16	2.1\\
-0.131860091011558	2.09585612969937\\
-0.263199790485039	2.0834408727604\\
-0.393500760630022	2.06280322653025\\
-0.522248763046195	2.03402463837013\\
-0.648935688187389	1.99721868421982\\
-0.773061560637824	1.95253062036533\\
-0.894136512286653	1.90013681017864\\
-1.0116827156136	1.84024402809211\\
-1.12523626945589	1.77308864355423\\
-1.23434902981419	1.69893568818739\\
-1.33859037847225	1.61807780982916\\
-1.43754892245025	1.53083411758496\\
-1.53083411758496	1.43754892245025\\
-1.61807780982916	1.33859037847225\\
-1.69893568818739	1.23434902981419\\
-1.77308864355423	1.12523626945589\\
-1.84024402809211	1.0116827156136\\
-1.90013681017864	0.894136512286653\\
-1.95253062036533	0.773061560637824\\
-1.99721868421982	0.64893568818739\\
-2.03402463837013	0.522248763046195\\
-2.06280322653025	0.393500760630022\\
-2.0834408727604	0.26319979048504\\
-2.09585612969937	0.131860091011559\\
-2.1	2.57175827820944e-16\\
-2.09585612969937	-0.131860091011558\\
-2.0834408727604	-0.263199790485039\\
-2.06280322653025	-0.393500760630022\\
-2.03402463837013	-0.522248763046196\\
-1.99721868421982	-0.64893568818739\\
-1.95253062036533	-0.773061560637824\\
-1.90013681017864	-0.894136512286653\\
-1.84024402809211	-1.0116827156136\\
-1.77308864355423	-1.12523626945589\\
-1.69893568818739	-1.23434902981419\\
-1.61807780982916	-1.33859037847225\\
-1.53083411758496	-1.43754892245025\\
-1.43754892245025	-1.53083411758496\\
-1.33859037847225	-1.61807780982916\\
-1.2343490298142	-1.69893568818739\\
-1.12523626945589	-1.77308864355423\\
-1.0116827156136	-1.84024402809211\\
-0.894136512286653	-1.90013681017864\\
-0.773061560637825	-1.95253062036533\\
-0.64893568818739	-1.99721868421982\\
-0.522248763046196	-2.03402463837013\\
-0.393500760630022	-2.06280322653025\\
-0.26319979048504	-2.0834408727604\\
-0.131860091011558	-2.09585612969937\\
-3.85763741731416e-16	-2.1\\
0.131860091011557	-2.09585612969937\\
0.263199790485039	-2.0834408727604\\
0.393500760630021	-2.06280322653025\\
0.522248763046195	-2.03402463837013\\
0.648935688187389	-1.99721868421982\\
0.773061560637824	-1.95253062036533\\
0.894136512286652	-1.90013681017864\\
1.0116827156136	-1.84024402809211\\
1.12523626945589	-1.77308864355423\\
1.23434902981419	-1.69893568818739\\
1.33859037847225	-1.61807780982916\\
1.43754892245025	-1.53083411758496\\
1.53083411758496	-1.43754892245025\\
1.61807780982916	-1.33859037847225\\
1.69893568818739	-1.23434902981419\\
1.77308864355423	-1.12523626945589\\
1.84024402809211	-1.0116827156136\\
1.90013681017864	-0.894136512286653\\
1.95253062036533	-0.773061560637825\\
1.99721868421982	-0.64893568818739\\
2.03402463837013	-0.522248763046196\\
2.06280322653025	-0.393500760630022\\
2.0834408727604	-0.26319979048504\\
2.09585612969937	-0.131860091011558\\
2.1	-5.14351655641888e-16\\
};
\addplot [color=black, line width=2.0pt, forget plot]
  table[row sep=crcr]{%
2.6	0\\
2.59486949391351	0.163255350776215\\
2.57949822341764	0.325866407267191\\
2.55394685189459	0.487191417922884\\
2.51831621893444	0.646593706628622\\
2.4727469423674	0.803444185374863\\
2.41741886330945	0.957123836980163\\
2.35255033641165	1.10702615806919\\
2.27839736811405	1.25255955266446\\
2.19525260630524	1.39314966694539\\
2.10344418537486	1.52824165596043\\
2.00333443121705	1.65730237334659\\
1.89531843129567	1.77982247541459\\
1.77982247541459	1.89531843129567\\
1.65730237334659	2.00333443121705\\
1.52824165596043	2.10344418537486\\
1.39314966694539	2.19525260630524\\
1.25255955266446	2.27839736811405\\
1.10702615806919	2.35255033641165\\
0.957123836980163	2.41741886330945\\
0.803444185374863	2.4727469423674\\
0.646593706628623	2.51831621893444\\
0.487191417922884	2.55394685189459\\
0.325866407267191	2.57949822341764\\
0.163255350776215	2.59486949391351\\
1.59204083889156e-16	2.6\\
-0.163255350776215	2.59486949391351\\
-0.325866407267191	2.57949822341764\\
-0.487191417922885	2.55394685189459\\
-0.646593706628622	2.51831621893444\\
-0.803444185374863	2.4727469423674\\
-0.957123836980163	2.41741886330945\\
-1.10702615806919	2.35255033641165\\
-1.25255955266446	2.27839736811405\\
-1.39314966694539	2.19525260630524\\
-1.52824165596043	2.10344418537486\\
-1.65730237334659	2.00333443121705\\
-1.77982247541459	1.89531843129567\\
-1.89531843129567	1.77982247541459\\
-2.00333443121705	1.65730237334659\\
-2.10344418537486	1.52824165596043\\
-2.19525260630524	1.39314966694539\\
-2.27839736811405	1.25255955266446\\
-2.35255033641165	1.10702615806919\\
-2.41741886330945	0.957123836980163\\
-2.4727469423674	0.803444185374864\\
-2.51831621893444	0.646593706628623\\
-2.55394685189459	0.487191417922884\\
-2.57949822341764	0.325866407267192\\
-2.59486949391351	0.163255350776215\\
-2.6	3.18408167778312e-16\\
-2.59486949391351	-0.163255350776215\\
-2.57949822341764	-0.325866407267191\\
-2.55394685189459	-0.487191417922884\\
-2.51831621893444	-0.646593706628623\\
-2.4727469423674	-0.803444185374864\\
-2.41741886330945	-0.957123836980164\\
-2.35255033641165	-1.10702615806919\\
-2.27839736811405	-1.25255955266446\\
-2.19525260630524	-1.39314966694539\\
-2.10344418537486	-1.52824165596043\\
-2.00333443121705	-1.65730237334659\\
-1.89531843129567	-1.77982247541459\\
-1.77982247541459	-1.89531843129567\\
-1.65730237334659	-2.00333443121705\\
-1.52824165596043	-2.10344418537486\\
-1.39314966694539	-2.19525260630524\\
-1.25255955266446	-2.27839736811404\\
-1.10702615806919	-2.35255033641165\\
-0.957123836980164	-2.41741886330945\\
-0.803444185374864	-2.4727469423674\\
-0.646593706628624	-2.51831621893444\\
-0.487191417922884	-2.55394685189459\\
-0.325866407267192	-2.57949822341764\\
-0.163255350776214	-2.59486949391351\\
-4.77612251667468e-16	-2.6\\
0.163255350776213	-2.59486949391351\\
0.325866407267191	-2.57949822341764\\
0.487191417922883	-2.55394685189459\\
0.646593706628623	-2.51831621893444\\
0.803444185374863	-2.4727469423674\\
0.957123836980163	-2.41741886330945\\
1.10702615806919	-2.35255033641165\\
1.25255955266446	-2.27839736811405\\
1.39314966694539	-2.19525260630524\\
1.52824165596043	-2.10344418537486\\
1.65730237334659	-2.00333443121705\\
1.77982247541459	-1.89531843129567\\
1.89531843129567	-1.77982247541459\\
2.00333443121705	-1.65730237334659\\
2.10344418537486	-1.52824165596043\\
2.19525260630524	-1.39314966694539\\
2.27839736811405	-1.25255955266446\\
2.35255033641165	-1.10702615806919\\
2.41741886330945	-0.957123836980165\\
2.4727469423674	-0.803444185374864\\
2.51831621893444	-0.646593706628624\\
2.55394685189459	-0.487191417922884\\
2.57949822341764	-0.325866407267192\\
2.59486949391351	-0.163255350776215\\
2.6	-6.36816335556624e-16\\
};
\end{axis}
\end{tikzpicture}%
\caption{An example of a support of an optimal input distribution for $n=2$.}
\label{fig:InputSphere}
\end{figure}
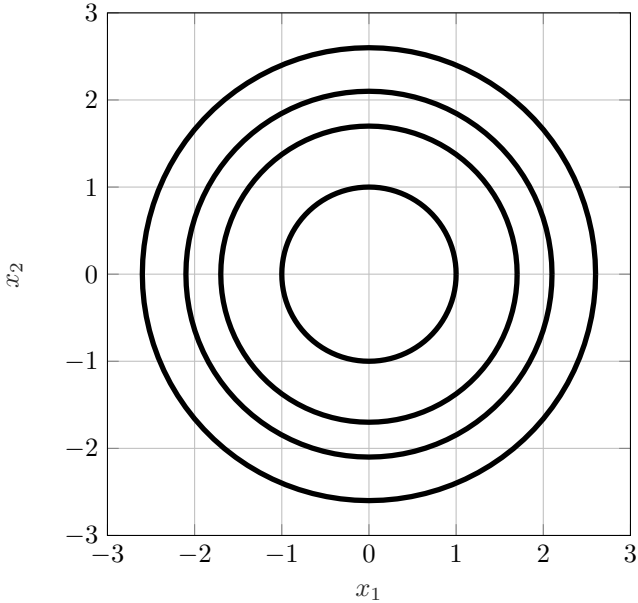

 Regrettably, the method of proof by contradiction does not lead to a characterization of the number of spheres (number of mass points when $n=1$) in the capacity-achieving input distribution.  In fact, as of the writing of this paper, very little is known about the structure of that distribution, and a very simple question remains open about $50$ years after Smith's contribution:
\begin{center}
\emph{When $n=1$, what is the cardinality of the support  of the optimal input distribution as a  function of  $\sfA$?} 
\end{center}
  In this work, we provide the first firm upper bound on the number of  points for $n=1$ and the number of shells for every $n> 1$, partially answering the above question.
    Furthermore, for the case of $n=1$, using similar methods, we also provide an upper bound on the cardinality of the support of the distribution achieving
  \begin{align}
    C(\sfA, \sfP) = \max_{\substack{X\colon |X|\le \sfA \\  \bbE[X^2]\le \sfP }} I(X;Y) \text{.}  \label{eq:channel power}
  \end{align}

\subsection{Prior Work}

 The history of the problem begins with  Shannon  who was the first to consider an amplitude constraint on the input \cite{Shannon:1948}. Shannon's original paper proposes both upper and lower bounds on the capacity and shows that peak-power capacity and average power capacity have the same asymptotic behavior at the low signal to noise ratio. The next major breakthrough is the seminal paper of Smith \cite{smith1971information}, where Smith proves the discreteness of the capacity-achieving input distribution and also shows the optimality of the equiprobable binary input on $\{ \pm \sfA \}$ so long as  $\sfA \le 0.1$. Sharma and Shamai \cite{sharma2010transition} extend the result of Smith, and argue\footnote{The formal proof is incomplete as the argument in \cite{sharma2010transition} uses a conjecture which presumes that the number of points in the optimal distribution increases by 1 as the amplitude constraint is relaxed. Proof of this fact was later shown to be true in \cite{dytso2019capacityEstim}.} that an equiprobable input on $\{ \pm \sfA \}$ is optimal if and only if  $\sfA \le \bar{\sfA} \approx 1.665$.  The proof of the result in  \cite{sharma2010transition}, which generalizes to vector channels, is shown in \cite{dytso2019capacityEstim}.   Also, based on numerical evidence, Sharma and Shamai \cite{sharma2010transition} conjectured that the number of mass points increases by at most one and a new point always appears at zero.  Based on this conjecture, in \cite{sharma2010transition} it has been shown that  a ternary input distributed on $\{ -\sfA,0,\sfA \}$ is optimal for  all $ \bar{\sfA} \le  \sfA \le \bar{\bar{\sfA}}  \approx 2.786$.

A progress on the algorithmic aspect of computing the optimal input distribution was made in \cite{huang2005characterization} which proposed an iterative procedure  that converges to the a capacity achieving distribution based on the cutting-plane method. The bound on the number of mass points found in our work is particularly relevant for numerical methods as it reduces the optimization space for algorithms such as the one contained in \cite{huang2005characterization}.  

A number of papers have also focused on upper and lower bounds on the capacity in \eqref{eq:Capacity}. Broadly speaking, there are three types of capacity upper bounding approaches. The first approach uses   the maximum entropy principle \cite[Chapter 12]{CoverInfoTheory} and upper bounds the output differential entropy, $h(Y)$, subject to some moment constraint   \cite{dytso2017amplitude_Globecom}. The second approach uses a dual capacity characterization\footnote{Also known as Redundancy-Capacity Theorem (see, e.g., \cite{gallagher1976}, \cite{ryabko1979}.)} where the maximization of the mutual information  over the input distribution is replaced by minimization of the relative entropy over the output distribution. A suboptimal choice of an output distribution in the dual capacity expression results in an upper bound on the capacity \cite{mckellips2004simple,thangaraj2017capacity,lapidoth2003capacity,rassouli2016upper}. 
 The third approach uses a characterization of the mutual information as an integral of the minimum mean square error (MMSE) \cite{I-MMSE}, and  leads to an upper bound by replacing the optimal estimator in the MMSE term by a suboptimal one \cite{dytso2019capacityEstim}. As for the lower bounds on the capacity, the first one relevant to our setting, as mentioned above, was proposed by Shanon in \cite{Shannon:1948} which was based on the entropy power inequality. Other important lower bounds include Ozarow-Wyner bounds \cite{ozarow1990capacity, GenOWbounds}, and bounds based on Jensen's inequality \cite{Entropy2019}.  

There is also a substantial literature that extends the proof recipe of Smith to the other channels.  For example, the  approach of Smith for showing discreteness of an optimal input distribution has been extended to complex Gaussian channels \cite{ShamQuadrat}, additive  noise channels where noise has a sufficiently regular pdf \cite{tchamkerten2004discreteness},  Rayleigh fading channels \cite{abou2001capacity}, and Poisson channels \cite{shamai1990capacity}.  For an overview of the literature on various optimization methods that show discreteness of a capacity-achieving distribution the interested reader is referred to \cite{CISS2018}. Moreover, a comprehensive account of capacity results for point-to-point Gaussian channels can be found in \cite{verdu1998fifty}. 
  
  One of the ingredients of our proof is the Oscillation Theorem of Karlin \cite{karlin1957polya}.  In the past, Karlin's theorem has been used to study extreme distributions; however, not to the same degree as it is used in this paper. For example,  in the context of a Bayesian estimation problem \cite{casella1981estimating}, Oscillation Theorem has been used to show the necessary and sufficient conditions for a binary distribution to be the least favorable. In \cite{dytso2019capacityEstim}, in a vector version of the optimization in \eqref{eq:Capacity},  Oscillation Theorem has been used to show the necessary and  sufficient conditions for a uniform distribution on a single sphere to be optimal. 

\subsection{Contributions and Paper Outline}
In what follows:
\begin{enumerate}
\item Section~\ref{sec:MainResult} presents our main results;
\item Section~\ref{sec:ProofMainResult} provides the proof of the first part of our main result for the case of $n=1$. 
 There, it is shown that the number of zeros of the shifted optimal output probability density function (pdf) is within a factor of two from  the number of mass points of the optimal input distribution. The main element of this part relies on Karlin's Oscillation Theorem; 

\item  Section~\ref{sec:BoundOnTheNumberOfOscillations} provides the proof of the second part of the main result for the case of $n=1$. Specifically, an explicit upper bound on the number of extreme points of an arbitrary output pdf of the  Gaussian channel described in \eqref{eq:channel} is derived. The proof of this result exploits the analyticity of the Gaussian density together with Tijdeman's Number of Zeros Lemma \cite[Lemma 1]{Tijdeman1971number}.  The proof for the vector case ($n > 1$)  follows along the same lines as the proof for the scalar case ($n=1$), albeit with a more involved algebra, therefore it is relegated to the Appendix;  and
\item  Section~\ref{sec:Conclusion} concludes the paper with some final remarks. 
 \end{enumerate}
\subsection{Notation} 
Throughout the paper, the deterministic scalar quantities are denoted by lower-case letters, deterministic vectors are denoted by bold lowercase letters, random variables are denoted by uppercase letters, and random vectors are denoted by bold uppercase  (e.g., $x$, $\boldsymbol{x}$, $X$, $\boldsymbol{X}$).
%
 We denote the distribution  of a random vector $\boldsymbol{X}$ by $P_{\boldsymbol{X}}$. Moreover, we say that a point $\boldsymbol{x}$ is in the support, denoted by $\supp(P_{\boldsymbol{X}})$,\footnote{Also known as ``points of increase of $P_{\boldsymbol{X}}$" or ``spectrum of $P_{\boldsymbol{X}}$."} of the distribution $P_{\boldsymbol{X}}$ if for every open set $ \mathcal{O} \ni \boldsymbol{x} $ we have that $P_{\boldsymbol{X}}(\cO)>0$. We refer to symmetric random variables as those that  are symmetric with respect to the origin. 

The number of zeros of a function $f \colon \mathbb{R} \to \mathbb{R} $  on the interval $\cI$ is denoted by\footnote{The definition $\rmN(\cI, f)$ is blind to the multiplicities of the zeros.}  $\rmN(\cI, f)$. Similarly, if $f  \colon \bbC \to \bbC$ is a function on the complex domain, $\rmN(\cD, f)$ denotes the number of its  zeros within the region $\cD$.

Finally, while the relative entropy between $X$ and $Y$ is denoted by $D(X\| Y)$, the entropy of a discrete random variable $X$ is denoted by $H(X)$ and the differential entropy of a continuous random variable $X$ is denoted by $h(X)$. 

\section{ Main Results} \label{sec:MainResult}

Theorem~\ref{thm:MainResult}, stated below, gives the first firm upper bound on the support size of the capacity-achieving input of the scalar additive Gaussian channel with an amplitude constraint.

\begin{theorem}\label{thm:MainResult} 
Consider the amplitude constrained scalar additive Gaussian channel $Y= X+ Z $ where the input $X$, satifying $|X|\le \sfA $, is assumed to be independent from the noise $Z\sim \cN(0,1)$. Assuming $\sfA \ge 1$, let $P_{X^\star}$ be the optimizing input distribution for this channel. Then,   $P_{X^\star}$ is a symmetric discrete distribution with
\begin{align}
  \frac{1}{2} \rmN\left([-R, R], f_{Y^\star}  -  \kappa_1 \right)   &\le |\supp(P_{X^\star})| \label{eq:BoundViaOutputPDF_lower} \\ 
  & \le \rmN\left([-R, R], f_{Y^\star}  -  \kappa_1 \right) \label{eq:BoundViaOutputPDF_upper} \\ 
  & <\infty \text{,} 
\end{align}
where  $\kappa_1=\rme^{-C(\sfA)-h(Z) } $   and\footnote{Unless otherwise stated, the logarithms in this paper are of base $\rme$.}   $R=  \sfA + \log^{\frac{1}{2}} \left(\frac1{2 \pi \kappa_1^2}\right )$.   Moreover,  
\begin{align}
\sqrt{1+\frac{2 \sfA^2 }{\pi \rme } } &\le |\supp(P_{X^\star})|  \label{eq:EPI_lowbd} \\ 
&\le \rmN\left([-R, R], f_{Y^\star}  -  \kappa_1 \right) \\ 
&\le  a_2 \sfA^2 +a_1 \sfA+ a_0 \text{,} \label{eq:MainBound} 
\end{align}
with 
\begin{align}
a_2&=9\rme+ 6\sqrt{\rme}+5\text{,}\\
a_1&=  6\rme+2\sqrt{\rme} \text{,} \\
a_0&=\rme+2\log\left(4 \sqrt{\rme}+2\right)+1 \text{.}
\end{align}
\end{theorem}

Since it consists of two parts, the proof of Theorem~\ref{thm:MainResult} is divided into two sections. While Section~\ref{sec:ProofMainResult} proves the order tight bounds \eqref{eq:BoundViaOutputPDF_lower} and \eqref{eq:BoundViaOutputPDF_upper}, Section~\ref{sec:BoundOnTheNumberOfOscillations} finds the lower and upper bounds presented in \eqref{eq:EPI_lowbd} and \eqref{eq:MainBound}.

\begin{rem}
Observe that the bounds in \eqref{eq:BoundViaOutputPDF_lower} and \eqref{eq:BoundViaOutputPDF_upper} are order tight. While the same cannot be said about the bounds in \eqref{eq:EPI_lowbd} and \eqref{eq:MainBound}, we conjecture that the order of the lower bound in \eqref{eq:EPI_lowbd} is the one that is tight. A possible approach for tightening the upper bound  is discussed in Section~\ref{sec:BoundOnTheNumberOfOscillations} along with a figure that supports our conjecture, see Figure~\ref{fig:NumberOfZeros}. 
\end{rem}

\begin{theorem}\label{thm:vector Gaussian}
Consider the amplitude constrained vector additive Gaussian channel $\bsY=\bsX+\bsZ $ where the input $\bsX $, satisfying $ \|\bsX\| \le \sfA$, is assumed to be independent from the white Gaussian noise $\bsZ\sim \cN(\boldsymbol{0},\bfI_{n})$. Let $\bsX^\star \sim P_{\bsX^\star}$ be the optimizing input for this channel. Then, $P_{\bsX^\star}$ is unique, radially symmetric, and the distribution of its amplitude, namely $P_{\|\bsX^\star \|} $, is a discrete distribution with
\begin{align}
 \left|\supp\left(P_{\|\bsX^\star\|}\right)\right| \le \mfra_{n_2} \sfA^2 + \mfra_{n_1} \sfA +\mfra_{n_0} \text{,}
\end{align}
where, denoting the gamma function by $\opGamma$,
\begin{align}
  \mfra_{n_2} &=  4+4\rme +\sqrt{8\rme +4}  \text{,} \label{eqn:def:a_n_2Vectorcase}  \\
   \mfra_{n_1} &= \left(3+4\rme + \sqrt{2\rme +1}\right)n + \sqrt{\frac{32}{n-1}} \text{,} \label{eqn:def:a_n_1Vectorcase} \\ 
   \mfra_{n_0} &= \log \frac{\rme^2 \sqrt{\pi}\opGamma\left( \frac n2\right) }{\opGamma\left(\frac{n-1}{2}\right)} \notag  \\ 
    & + \left(3+4\rme + \sqrt{2\rme +1}\right) \left(\frac n2 +\log \frac{\sqrt{\pi}\opGamma\left(\frac n2\right) }{\opGamma\left(\frac{n-1}{2}\right)}\right) \text{.} \label{eqn:def:a_n_0Vectorcase}
\end{align}
\end{theorem}

The proof of Theorem~\ref{thm:vector Gaussian} benefits from the same technique that is used in the proof of Theorem~\ref{thm:MainResult}. For this reason, its presentation is postponed to Appendix~\ref{app:thm:thm:vector Gaussian}.

 \begin{rem} Note that when the vector channel is of dimension $2$, Theorem~\ref{thm:vector Gaussian} gives an upper bound on the number of shells of the optimal input distribution for the additive complex Gaussian channel with an amplitude constraint. 
 \end{rem}

For the sake of demonstrating the versatility of our novel method, proven next is an upper bound on the support size of the optimal input distribution for the scalar additive Gaussian channel with both an amplitude and a power constraint.

\begin{theorem}\label{thm:AGC amplitude and power constraint}

Consider the amplitude and power constrained scalar additive Gaussian channel $Y= X+ Z $ where the input $X$, satisfying $|X|\le \sfA $ and $\bbE[|X|^2]\le \sfP $, is assumed to be independent from the noise $Z\sim \cN(0,1)$. Assuming $\sfA \ge 1$, let $P_{X^\star}$ be the optimizing input distribution for this channel. Then,   $P_{X^\star}$ is a symmetric discrete distribution with
\begin{align}
\sqrt{1+ \frac{ 2 \min \left\{\sfA^2, 3 \sfP \right\}}{\pi \rme}}  &\le |\supp(P_{X^\star})|  \\
&\le a^{}_{\sfP_2} \sfA_\sfP^2 + a^{}_{\sfP_1} \sfA_\sfP + a^{}_{\sfP_0}  \text{,}
\end{align}
where
\begin{align} 
\sfA_\sfP &= \frac{\sfA \sfP}{\sfP - \log(1+\sfP) 1\left\{\sfP < \sfA^2\right\}} \text{,} \label{eqn:def:A_P}  \\
  a^{}_{\sfP_2} &= (1+2\lambda_\sfP )(9\rme+ 6\sqrt{\rme}+1) + 2(2-\lambda_\sfP)(1-2\lambda_\sfP )  \text{,}  \label{eqn:def:a_{sfP_2}} \\
   a^{}_{\sfP_1} &= (1+2\lambda_\sfP )(6\rme + 2\sqrt{\rme} )   \text{,} \label{eqn:def:a_{sfP_1}} \\
    a^{}_{\sfP_0} &= (1+2\lambda_\sfP )\rme + 2\log\left( \frac{2+4\sqrt{\rme}(1+2\lambda_\sfP ) }{1-2\lambda_\sfP } \right)  +1 \text{,} \label{eqn:def:a_{sfP_0}} \\
    \lambda_\sfP &= \frac{\log(1+ \sfP)}{2\sfP}  \cdot 1\left\{\sfP < \sfA^2\right\} \text{.} \label{eqn:def:lambda_P}
\end{align}
\end{theorem}

With only small alterations, proof of Theorem~\ref{thm:AGC amplitude and power constraint} imitates that of Theorem~\ref{thm:MainResult} and is shown in Appendix~\ref{apdx:AGC amplitude and power constraint}.

\begin{rem}
  In the case when $\sfP \ge \sfA^2 $, the power constraint becomes inactive and Theorem~\ref{thm:AGC amplitude and power constraint} recovers the result of Theorem~\ref{thm:MainResult}.
\end{rem}

\section{ Proof for the First Part of Theorem~\ref{thm:MainResult}} 
\label{sec:ProofMainResult}

This section proves the first part of our main result in Theorem~\ref{thm:MainResult}, namely the bounds in \eqref{eq:BoundViaOutputPDF_lower} and \eqref{eq:BoundViaOutputPDF_upper}.

\subsection{On Equations Characterizing the Support of $P_{X^\star}$  } 

The first ingredient of the proof is  the following characterization of the optimal input distribution shown in \cite[Corollary 1]{smith1971information}.
\begin{lemma}\label{lem:SmithResult} 
 Consider the amplitude constrained scalar additive Gaussian channel $Y = X+Z$ where the input $X$, satisfying $|X|\le \sfA $, is independent from the noise $Z\sim \cN(0,1)$. Then, $P_{X^\star} $ is the capacity-achieving input distribution if and only if the following two equations are satisfied:
\begin{align}
i(x;P_{X^\star}) &= C(\sfA),   \quad   x \in  \supp(P_{X^\star}) \text{,}\\
i(x;P_{X^\star})  &\le  C(\sfA),   \quad   x \in  [-\sfA, \sfA] \text{,}
\end{align} 
where $C(\sfA)$ denotes the capacity of the channel, and 
\begin{align}
i(x;P_{X^\star})=  \int_{\bbR}   \frac{ \rme^{-\frac{(y-x)^2}{2}}}{\sqrt{2 \pi}}  \log  \frac{1}{f_{Y^\star}(y)} \rmd y  - h(Z)\text{,} \label{eqn:lem:def:i(x;P_X^star)}
\end{align} 
with $h(Z) =  \log \sqrt{2\pi \rme} $ denoting the differential entropy of the standard Gaussian distribution, and $f_{Y^\star}(y)$ denoting the output pdf induced by the input $P_{X^\star}$, that is, for $X\sim P_{X^\star} $,
\begin{align}
f_{Y^\star}(y)= \frac{1}{\sqrt{2 \pi}} \bbE \left[ \rme^{-\frac{(y-X)^2}{2}} \right] \text{.} 
\end{align} 
\end{lemma}

\begin{rem} \label{rem:RemarkOnSmith's_Result}
An immediate consequence of Lemma~\ref{lem:SmithResult} is the fact that 
\begin{align*}
 \supp(P_{X^\star}) 
 &\subseteq  \left\{x :    i(x;P_{X^\star}) - C(\sfA) = 0  \right \} \text{,} 
\end{align*} 
which leads to the following inequalities on the size of the support: 
\begin{align}
  |\supp(P_{X^\star})| &\le \rmN([-\sfA,\sfA], \Xi_\sfA(\cdot ;P_{X^\star})) \label{eq:not_informative} \\
  &\le \rmN(\bbR, \Xi_\sfA(\cdot ;P_{X^\star})) \text{,}
\end{align}
where the function $\Xi(\cdot ;P_{X^\star}) \colon \bbR \to \bbR  $ is defined as
\begin{align}
\Xi_\sfA(x;P_{X^\star})= i(x;P_{X^\star}) -C(\sfA)\text{.} \label{eqn:def:Xi_A(x;P_X)}
\end{align}
Note that, as it stands, the upper bound in \eqref{eq:not_informative} does not yet reveal any information on the discreteness of $P_{X^\star}$ as the right side might just as well be $\infty$.
\end{rem}

\subsection{Connecting the Number of Oscillations of $f_{Y^\star}$ to the Number of Masses in $P_{X^\star}$ }

This section gives an alternative proof that $P_{X^\star}$ is discrete by relating the cardinality of $ \supp(P_{X^\star})$ to the number of zeros of the shifted output pdf $f_{Y^\star} - \rme^{-C(\sfA)-h(Z)} $.  The following definition sets the stage. 

\begin{definition}[Sign Changes of a Function]  The number of sign changes of a function $\xi$ is given by 
\begin{align}
  \scrS(\xi) = \sup_{m\in \bbN } \left\{\sup_{y_1< \cdots< y_m} \scrN \{ \xi (y_i) \}_{i=1}^m\right\} \text{,}
\end{align}
where  $\scrN\{ \xi (y_i) \}_{i=1}^m$ is the number of changes of sign of the sequence $\{ \xi (y_i) \}_{i=1}^m $.

\end{definition} 

 Proven in \cite{karlin1957polya}, the following theorem is the main tool in connecting the number of zeros of a shifted output pdf $f_{Y^\star}$ to the number of mass points of a capacity-achieving input distribution $P_{X^\star}$.

\begin{theorem}[Oscillation Theorem \cite{karlin1957polya}]\label{thm:OscillationThoerem} Given open intervals $\bbI_1 $ and $\bbI_2$, let $p\colon \bbI_1\times \bbI_2  \to \bbR$ be a strictly totally positive kernel.\footnote{A function $f:\bbI_1 \times \bbI_2 \to \bbR$ is said to be strictly totally positive kernel of order $n$ if $\det\left([f(x_i,y_j)]_{i,j = 1}^{m}\right) >0 $ for all $1\le m \le n $, and for all $x_1< \cdots < x_m \in \bbI_1  $, and $y_1< \cdots < y_m \in \bbI_2$. If $f$ is strictly totally positive kernel of order $n$ for all $n\in \bbN$, then $f$ is called a strictly totally positive kernel.} For an arbitrary $y$, suppose $p(\cdot, y)\colon \bbI_1 \to \bbR $ is an $n$-times differentiable function. Assume that $\mu$ is a measure on $\bbI_2 $, and let $\xi \colon \bbI_2 \to \bbR $ be a function with $\scrS(\xi) = n$. For $x\in \bbI_1$, define
\begin{align}
\Xi(x)=  \int  \xi (y) p(x ,y) {\rm d} \mu(y) \text{.}
\end{align}
If $\Xi \colon \bbI_1 \to \bbR$ is an $n$-times differentiable function, then either $\rmN(\bbI_1, \Xi) \le n$, or $\Xi\equiv 0$.  
\end{theorem} 

Note that Theorem~\ref{thm:OscillationThoerem} is applicable in our setting as the Gaussian distribution is a strictly totally positive kernel \cite{karlin1957polya}. The following result shows the connection between the support size of $P_{X^\star}$ and the number of zeros of the shifted optimal output pdf $f_{Y^\star}$ and recovers the bounds in \eqref{eq:BoundViaOutputPDF_lower} and \eqref{eq:BoundViaOutputPDF_upper}.

\begin{theorem}\label{thm:BoundingSupportWithOscillations} 
The support set of the capacity-achieving input distribution  $P_{X^\star}$ satisfies
\begin{align}
\frac{1}{2} \rmN\left([-R, R], f_{Y^\star}  -  \kappa_1 \right)    &\le |\supp(P_{X^\star})| \label{lemmaeggregium:lowerbd} \\
&\le \rmN\left([-R, R], f_{Y^\star}  -  \kappa_1 \right) \label{lemmaeggregium:upperbd} \\ 
&<\infty \text{,}  \label{eq:SupporNumberOfZerosOfGaussian}
\end{align} 
where $\kappa_1=\rme^{-C(\sfA)-h(Z)}$ and\footnote{See Remark~\ref{rem:kappa's choice} and observe that $\kappa_1 \in \big(0, \frac1{\sqrt{2\pi}} \big)$.} $R>  \sfA + \log^{\frac{1}{2}} \left(\frac1{2 \pi \kappa_1^2}\right )$.
\end{theorem}

\begin{proof} 
To see \eqref{lemmaeggregium:upperbd} and \eqref{eq:SupporNumberOfZerosOfGaussian}, observe that $\Xi_\sfA(x; P_{X^\star})$, defined in \eqref{eqn:def:Xi_A(x;P_X)}, can be written as follows:
\begin{align}
\Xi_\sfA(x;P_{X^\star})=   \int_{\bbR}  \frac{\xi_\sfA(y)}{\sqrt{2 \pi}} \rme^{-\frac{(y-x)^2}{2}} \rmd y \text{,}
\end{align} 
where
\begin{align}
\xi_\sfA(y)=\log  \frac{1}{f_{Y^\star}(y)} - C(\sfA) - h(Z). 
\end{align}
   First, observe that  it is impossible for  $\Xi_\sfA(x;P_{X^\star}) =0$ for all $x \in \mathbb{R}$ since  otherwise   $\xi_\sfA(y)$ would be zero for all $y  \in \mathbb{R}$. 
Furthermore,  using the fact that the Gaussian distribution is a strictly totally positive kernel,
\begin{align}
|\supp(P_{X^\star})|  &\le  \rmN(\bbR, \Xi_\sfA(\cdot ;P_{X^\star})) \label{frm:lemma 6} \\
&\le \scrS(\xi_\sfA) \label{eq:ApplicationOfTheOscillationTheorem}\\
&\le \rmN(\bbR,  \xi_\sfA )  \label{sign change vs no of zeros}  \\
&= \rmN\left(\bbR,  f_{Y^\star}  -  \kappa_1   \right) \label{eq:ZerosOf_xi_A(y)}\\
&= \rmN\left([-R, R] ,  f_{Y^\star}  -  \kappa_1   \right) \label{eq:ApplytingConcetrationBoundOFZeros} \\
&<\infty \text{,} \label{eq:ApplytingConcetrationBoundOFZeros2}
\end{align} 
where \eqref{frm:lemma 6} is a consequence of Lemma~\ref{lem:SmithResult} (see Remark~\ref{rem:RemarkOnSmith's_Result}); \eqref{eq:ApplicationOfTheOscillationTheorem} follows from  Theorem~\ref{thm:OscillationThoerem}; \eqref{sign change vs no of zeros} follows because the number of zeros is an upper bound on the number of sign changes; \eqref{eq:ZerosOf_xi_A(y)} follows by observing that $\xi_\sfA(y)=0$ if and only if  $ f_{Y^\star}(y) = \kappa_1$; and finally \eqref{eq:ApplytingConcetrationBoundOFZeros} and \eqref{eq:ApplytingConcetrationBoundOFZeros2} follow from Lemma~\ref{lem:LocationOFZeros_RealCase}  in Section~\ref{sec:BoundOnTheNumberOfOscillations}. 

 To see \eqref{lemmaeggregium:lowerbd}, using the fact that $\supp(P_{X^\star})$ is a finite set, suppose $|\supp(P_{X^\star}) | = n $, let $x_1< \ldots< x_n $ the elements of $\supp(P_{X^\star}) $, and write
\begin{align}
 f_{Y^\star}(y)  &= \frac{1}{\sqrt{2 \pi}}\sum_{i=1}^{n} \hspace{-0.03cm} P_{X^\star}( x_i )     \exp \left( \hspace{-0.03cm}- \frac12 (y-x_i)^2\right)   \text{,}
\end{align}
where, by the definition of $\supp(P_{X^\star}) $, the probabilities satisfy $P_{X^\star}( x_i ) > 0  $ for each $i= 1, \ldots, n$. Observing that the number of zeros of $f_{Y^\star}(y)  -  \kappa_1$ is the same as the number of zeros of the right-shifted function $ f_{Y^\star}(y-|x_1|-1)  -  \kappa_1 $, let
\begin{align}
  f(y) &= f_{Y^\star}(y-|x_1|-1)  -  \kappa_1 \\ 
 &= \sum_{i=1}^{n}  a_i    \exp \left(- \frac12 (y-u_i)^2\right) -a_0 \text{,}
\end{align}
where for $i= 0, 1, \ldots, n $ both $a_i>0 $ and $u_i>0 $ as
\begin{align}
u_i &=  x_i + |x_1| + 1  \quad \text{for } i = 1, \ldots n  \text{,} \\
  a_i  &=\begin{cases}
    \kappa_1 & i=0 \\
    \frac1{\sqrt{2\pi}}P_{X^\star}( x_i )  & i =1, \ldots, n \text{.}
  \end{cases}
\end{align}
Given arbitrary $0< \epsilon_1 <  \cdots <\epsilon_n $, consider the perturbed function
\begin{align}
  \widetilde f(y, \epsilon_1, \ldots, \epsilon_n ) = \sum_{i=1}^n a_i \exp \left(-\frac12(1+\epsilon_i)(y-u_i)^2 \right) - a_0 \text{.}
\end{align}
Note that
  \begin{align}
    \rme^{-\frac12 y^2} \widetilde f(y, \epsilon_1, \ldots, \epsilon_n ) = \sum_{i=0}^n b_i \exp\left(-\frac{(2+\epsilon_i)}{2}(y-v_i)^2 \right) \text{,}
  \end{align}
  where 
\begin{align}
  \epsilon_0 &= -1 \text{,} \\
    b_i &=  
  \begin{cases}
  -a_0  & i=0 \text{,}  \\
  a_i \exp \left(-\frac{1+\epsilon_i}{2(2+\epsilon_i) } u_i^2 \right)  &  i = 1, \ldots, n \text{,}
  \end{cases} \\
   v_i &= \begin{cases}
    0  & i=0 \text{,} \\
    \frac{1+\epsilon_i }{2+\epsilon_i } u_i & i = 1, \ldots, n \text{,}
   \end{cases} 
  \end{align}
is  a linear combination of $n+1$ distinct Gaussians with distinct variances and therefore has at most $2n $ zeros \cite[Proposition 7]{kalai2010efficiently}. Since this holds for any arbitrary choice of $\epsilon_i$'s and since $ \widetilde f(y, \epsilon_1, \ldots, \epsilon_n ) \to f(y) $ as $ ( \epsilon_1, \ldots, \epsilon_n ) \to (0,\ldots, 0)$, it follows that
\begin{align}
2|\supp(P_{X^\star}) | &\ge \rmN(\bbR, f(y)) \\
 &= \rmN(\bbR, f_{Y^\star}(y)  -  \kappa_1  )   \\
 &=  \rmN([R, R], f_{Y^\star}(y)  -  \kappa_1  ) \text{.}
\end{align} 
\end{proof}

\begin{rem}
With a different approach than the one taken in \cite{smith1971information}, observe that Theorem~\ref{thm:BoundingSupportWithOscillations} recovers the result of Smith \cite{smith1971information} showing that the support set of $P_{X^\star}$ is finite; hence, $P_{X^\star}$ is discrete with finitely many mass points. An advantage of the proof presented here is that the Fourier analysis required in the proof provided by \cite{smith1971information} is now completely avoided. Another advantage is that, since the presented proof is of the constructive nature, one can indeed attempt at counting the zeros of $f_{Y^\star}  -  \kappa_1 $, which is the topic of the next section. 
\end{rem}

\begin{rem} Theorem~\ref{thm:BoundingSupportWithOscillations} proves that the number of zeros of the shifted optimal output pdf $f_{Y^\star}  -  \kappa_1 $ gives an order tight upper bound on the support size of the optimal input pmf $|\supp(P_{X^\star})|$. This result can be considered as the main result of this paper.  
\end{rem}

\section{Proof of the Explicit Bounds in \eqref{eq:EPI_lowbd} and \eqref{eq:MainBound}} 
\label{sec:BoundOnTheNumberOfOscillations}

Section~\ref{sec:ProofMainResult} demonstrates that the number of mass points of $P_{X^\star}$ is  within a factor of two of the number of zeros of $f_{Y^\star}  -  \kappa_1 $ where $\kappa_1=\rme^{-C(\sfA)-h(Z) } $. In this section, we first provide an upper bound on the number of zeros of  $f_{Y^\star}  -  \kappa_1 $ and establish \eqref{eq:MainBound}. Additionally, through the use of entropy-power inequality, we also provide a lower bound on the support size of $P_{X^\star}$, yielding \eqref{eq:EPI_lowbd}. 

\begin{rem}
  A critical observation here is that, due to the lack of knowledge of the optimal input distribution $P_{X^\star}$ or the capacity expression $C(A)$, we do not know the optimal output distribution $f_{Y^\star} $ nor the constant $\kappa_1 = \rme^{-C(A)-h(Z)}$. Therefore, we must instead work with generic  $f_Y$ and $\kappa_1$ throughout this section.
\end{rem}

 \subsection{Bounds on the Number of Extreme Points of a Gaussian Convolution}
The aim of this subsection of the paper is to study the following problem: given an \emph{unknown constant} $ 0 \le  \kappa_1 \le \max_{b} f_Y(b) $,   find a worst-case upper bound on the number of zeros of the shifted output pdf
\begin{align*}
f_Y-\kappa_1 \text{,}
\end{align*} 
where $f_Y$ denotes the pdf of the random variable $Y=X+Z$, with $X$ being an \emph{arbitrary} zero mean\footnote{Since the channel is symmetric, the capacity-achieving input is symmetric. Therefore, there is no loss of optimality in restricting attention to zero mean inputs.} random variable at the input of the channel satisfying the amplitude constraint: $|X| \le \sfA$; $Z$ being the standard Gaussian random variable independent from $X$; and $Y$ being the random variable induced by the input $X$ at the output of this additive Gaussian channel.


As a starting point, before chasing after the number of zeros of $f_Y -\kappa_1$, the following lemma shows that the zeros of $f_Y -\kappa_1$ are always contained on an interval that is only ``slightly" larger than $[-\sfA,\sfA]$.

\begin{lemma}[On the Location and Finiteness of Zeros]\label{lem:LocationOFZeros_RealCase} For a fixed $\kappa_1 \in \left(0, \frac{1}{\sqrt{2 \pi}} \right]$ there exists some $\sfB_{\kappa_1} = \sfB_{\kappa_1}(\sfA) <\infty$ such that 
\begin{align}
 \rmN\left(\bbR,f_Y-\kappa_1 \right) &= \rmN\left([-\sfB_{\kappa_1}, \sfB_{\kappa_1}],f_Y-\kappa_1 \right) \\
 &<\infty\text{.}
 \end{align} 
In other words, there are finitely many zeros of $f_Y(y)-\kappa_1$ all of which are  contained within the interval $ [-\sfB_{\kappa_1}, \sfB_{\kappa_1}]$.    Moreover, $\sfB_{\kappa_1}$ can be upper bounded as follows: 
\begin{align}
\sfB_{\kappa_1} \le \sfA+ \log^{\frac{1}{2}} \left(\frac{1}{2 \pi \kappa_1^2} \right)\text{.}\label{eq:BOundONTheNumberOfZeros}\end{align}
\end{lemma} 
\begin{proof}
Using the monotonicity of $\rme^{-u} $, for all $|y|>\sfA$,
\begin{align}
f_Y(y) &= \frac{1}{\sqrt{2 \pi}}\bbE \left[  \rme^{-\frac{(y-X)^2}{2}} \right]  \\
       &\le  \frac{1}{\sqrt{2 \pi}} \rme^{-\frac{(y-\sfA)^2}{2}}\text{.}  \label{frm:rght is dec}
\end{align} 
Since the right side of \eqref{frm:rght is dec} is a decreasing function for all $|y|>A $, it follows that 
\begin{align}
  f_Y(y) - \kappa_1 < 0 
\end{align}
for all
\begin{align}
  |y| > \sfA+ \log^{\frac{1}{2}} \left(\frac{1}{2 \pi \kappa_1^2} \right) \text{.}
\end{align}
This means that there exists $\sfB_{\kappa_1} $ satisfying \eqref{eq:BOundONTheNumberOfZeros} such that all zeros of $f_Y- \kappa_1$ are located within the interval $[-\sfB_{\kappa_1}, \sfB_{\kappa_1}]$.
 
 To see that there are finitely many zeros, recall the fact that a convolution with a Gaussian distribution preserves analyticity \cite[Proposition 8.10]{folland2013real}; hence $f_Y$ is an analytic function on $\bbR$. Standard methods (e.g., invoking Bolzano-Weierstrass Theorem and the Identity Theorem) yield the fact that analytic functions have finitely many zeros on a compact interval, which is the desired result. 
 \end{proof} 

Since the exact value of the constant $\kappa_1 $ is unknown, in counting the number of zeros of $f_Y - \kappa_1  $, a worst-case approach needs to be taken. In an attempt at doing so, the following elementary result from calculus provides a bound on the number of zeros of a function in terms of the number of its extreme points. As simple as it is, Lemma~\ref{lem:extreme points and oscillations} is one of the key steps in this paper. It states that, to find a bound on the number of zeros of $f_Y -\kappa_1$, it suffices to find a bound on that of $f'_Y$,  eliminating the dependence on the nuisance constant $\kappa_1$.

\begin{lemma}\label{lem:extreme points and oscillations}
Suppose that $f$ is continuous  on $[-R,R]$ and differentiable on $(-R,R)$. If $\rmN([-R,R],f) <\infty$, then 
\begin{align}
  \rmN([-R,R],f) \le   \rmN([-R,R],f ') + 1 \text{,}
\end{align}
where $f'$ denotes the derivative of $f$.  
\end{lemma}

\begin{proof}
Let $x_1< \ldots< x_{\rmn_{0}}$ denote the zeros of $f$. By Rolle's Theorem, each of the intervals $(x_i, x_{i+1}) $ for $i=1, \ldots, \rmn_0-1$ contains at least one extreme point.
\end{proof}

Thanks to Lemma~\ref{lem:extreme points and oscillations}, to upper bound the number of zeros of $f_Y - \kappa_1 $, all that is needed  is to find an upper bound on the number of zeros of the derivative of $f_Y$, namely
\begin{align}
 f'_Y (y) = \frac1{\sqrt{2\pi}} \bbE\left[(X-y)\exp\left(-\frac{(y-X)^2}{2}\right)\right] \text{.} \label{eqn:def:f}
\end{align}

At this point, there are several trajectories that one could follow. For example, using the fact that $f'_Y $ is an analytic function (cf. \cite[Appendix B]{smith1969Thesis}) and letting $\breve f'_Y$ denote its complex analytic extension,
\begin{align}
\rmN([-R,R], f'_Y) &\le  \inf_{\epsilon>0}  \rmN(\cD_{R+\epsilon},\breve f'_Y) \label{frm:f=0,brve f=0} \\
&=   \inf_{\epsilon>0}  \frac{1}{2 \pi i} \oint\limits_{\del \cD_{R+\epsilon} }  \frac{  \breve  f''_Y(z)}{ \breve  f'_Y(z)}  \rmd z  \label{eq:ArgumentPrinciple} \\
&\le  \inf_{\epsilon>0}  (R+\epsilon)  \max_{|z|=R+\epsilon} \left| \frac{  \breve  f''_Y(z)}{\breve  f'_Y(z)} \right| \text{,} \label{frm:ML bd}
\end{align}
where in \eqref{frm:f=0,brve f=0} $\cD_t\subset \bbC $ is an open disc\footnote{In fact, $\cD_R$ can be any open connected set that contains the interval $[-R,R]$. For example, a rectangle of width $2(R+\epsilon)$ and  arbitrary hight $2 \epsilon$ is a typical choice.} of radius $t$ centered at the origin and the inequality follows because  $f'_Y(y)=0 \implies \breve f'_Y(y) = 0 $; in \eqref{eq:ArgumentPrinciple} $\del \cD_t $ denotes the boundary of the disc $\cD_t$ and equality follows from Cauchy's argument principle (e.g., \cite[Corollary 10.9]{bak1982complex}); and finally \eqref{frm:ML bd} follows from the ML inequality for the contour integral \cite[Chapter~4.10]{bak1982complex}.

Unfortunately, due to the implicit definitions of the functions $f''_Y$ and $f'_Y$, the maximization of the ratio $\breve f''_Y/\breve f'_Y$ in the right side of \eqref{frm:ML bd} does not seem to have a tractable explicit solution. Luckily, there are alternative, more tractable methods that yield an explicit upper bound on the number of zeros of $\breve f'_Y $. The method used in this paper is based on Tijdeman's Number of Zeros Lemma, which is presented next.

\begin{lemma}[Tijdeman's Number of Zeros Lemma \cite{Tijdeman1971number}]\label{lem:number of zeros of analytic function}
 Let $R, s, t$ be positive numbers such that $s>1$. For the complex valued function $f\neq  0$ which is analytic on $|z|<(st+s+t)R$, its number of zeros $  \rmN(\cD_R,f)$ within the disk $\cD_R = \{z\colon |z|\le R\} $ satisfies
\begin{align}
 & \rmN(\cD_R,f) \notag\\
  & \le \frac{1}{\log s} \left(\log \max_{|z|\le (st+s+t)R } |f(z)|   -\log \max_{|z|\le tR} |f(z)|\right) \text{.} \label{eq:Tijdeman}
\end{align}
\end{lemma}

The following two lemmas find upper and lower bounds on absolute value of the complex analytic extension\footnote{The fact that the complex extension of $f_Y$, and hence that of $f'_Y $, is analytic on $\bbC$ is proven in \cite[Appendix B]{smith1969Thesis}.} of $f'_Y $ over a disc of finite radius  centered at the origin. 

\begin{lemma}\label{lem:upper bound on maxima over a disk}
Suppose $f'_Y\colon \bbR \to \bbR$ is as in \eqref{eqn:def:f} and let $\breve f'_Y\colon \bbC \to \bbC $ denote its complex extension. Then,
  \begin{align}
\max _{|z|\le \sfB} \left|\breve f'_Y(z)\right| \le \frac1{\sqrt{2\pi}}(\sfA+\sfB)\exp\left(\frac{\sfB^2}2\right) \text{.}
  \end{align}
\end{lemma}

\begin{proof}
Using the standard rectangular representation of a complex number, let $z= \xi + i \eta   \in \bbC$, 
  \begin{align}
& \max_{|z|\le \sfB} \left|\breve f'_Y(z)\right|  \notag\\
  &= \max_{|z|\le \sfB} \left\{\frac1{\sqrt{2\pi}} \left|\bbE\left[(X-z)\exp\left(-\frac{(z-X)^2}{2}\right)\right]\right|\right\} \\
 &\le \max_{|z|\le \sfB} \left\{\frac1{\sqrt{2\pi}} \bbE\left[|X-z|\left|\exp\left(-\frac{(z-X)^2}{2}\right)\right|\right]\right\}  \label{frm:jensen yensen} \\ 
 &= \max_{|z|\le \sfB} \left\{\frac1{\sqrt{2\pi}} \bbE\left[|X-z|\exp\left(\frac{\eta^2-(\xi -X)^2}{2} \right)\right]\right\} \\ 
 &\le \max_{|z|\le \sfB} \left\{\frac1{\sqrt{2\pi}} \bbE\left[(|X|+|z|)\exp\left(\frac{\eta^2}{2} \right)\right]\right\} \label{frm:trian and e to the pow} \\
 &\le \frac1{\sqrt{2\pi}}(\sfA+\sfB)\exp\left(\frac{\sfB^2}2\right) \text{,}  \label{frm:fin in lem}
  \end{align} 
where \eqref{frm:jensen yensen} follows from Jensen's inequality; \eqref{frm:trian and e to the pow} follows from triangle inequality; and finally \eqref{frm:fin in lem} is because $|X| \le \sfA$, and $|z|\le \sfB$ implies $|\eta|\le \sfB$. 
\end{proof}

\begin{lemma}\label{lem:lower bound on maxima over a disk}
Suppose $f'_Y\colon \bbR \to \bbR$ is as in \eqref{eqn:def:f} and let $\breve f'_Y\colon \bbC \to \bbC $ denote its complex extension. For any $ |X| \le \sfA \le \sfB $, 
\begin{align}
    \max_{|z|\le \sfB} \left|\breve f'_Y(z)\right| \ge \frac \sfA{\sqrt{2\pi}} \exp\left(-2\sfA^2\right) \text{.} 
\end{align}

\end{lemma}

\begin{proof}
  Thanks to the suboptimal choice of $z=\sfA\le \sfB$, 
  \begin{align}
    \max_{|z|\le \sfB} \left|\breve f'_Y(z)\right| &\ge \frac{1}{\sqrt{2\pi}} \left|\bbE\left[(X-\sfA)\exp\left(-\frac{(\sfA-X)^2}{2}\right)\right]\right| \\ 
    &\ge  \frac{1}{\sqrt{2\pi}} \bbE\left[(\sfA-X)\exp\left(-2\sfA^2\right)\right] \label{frm:X is bdd int rv}  \\
    &= \frac \sfA{\sqrt{2\pi}} \exp\left(-2\sfA^2\right) \text{,} \label{frm:mean of X is zero}
  \end{align}
  where \eqref{frm:X is bdd int rv} follows because $|X|\le \sfA $; and \eqref{frm:mean of X is zero} is a consequence of $\bbE[X]=0$. 
\end{proof} 
By assembling the results of Lemmas~\ref{lem:number of zeros of analytic function},~\ref{lem:upper bound on maxima over a disk}~and~\ref{lem:lower bound on maxima over a disk}, Theorem~\ref{thm:NumberOfOscillations} below provides an upper bound on the number of oscillations of  a Gaussian convolution.   
\begin{theorem}[Bound on the Number of Oscillations of $f_Y$]\label{thm:NumberOfOscillations}
Let  $|X|\le \sfA< R$  for some fixed $R$.  
Then, the number of extreme points of $f_Y$, namely the number of zeros of $f'_Y$,  within the interval $[-R, R]$ satisfies
  \begin{align}
 &\rmN([-R,R],f'_Y) \notag\\
 &  \le  \min_{s> 1}\left\{\frac{\left( \frac{((\sfA+R)s+\sfA)^2}2 + 2\sfA^2+ \log\left(2+\frac{(\sfA+R)s}{\sfA}\right) \right)}{\log s} \right\} \text{.}
  \end{align}
\end{theorem}

\begin{proof}
Let $\cD_R\subset \bbC$ be a disk of radius $R$ centered at $z_0=0$, and note that
  \begin{align}
& \rmN([-R,R],f'_Y)  \notag\\
 &\le \rmN(\cD_R,\breve f'_Y) \label{frm:zeros in real zeros in complex} \\
     &\le  \min_{s> 1,\, t\ge  \frac{\sfA}{R}} \left\{ \frac{\log  \frac{\max_{|z| \le (st+s+t)R}|\breve f'_Y (z)| }{\max_{|z| \le tR} |\breve f'_Y (z)|}}{\log s}  \right\} \label{frm:lem: bd on the no of zeros of analytic func} \\
     &\le  \min_{s> 1,\, t\ge  \frac{\sfA}{R}}\left\{\frac{\frac{(st+s+t)^2R^2}2 + 2\sfA^2+ \log\left(1+\frac{(st+s+t)R}{\sfA}\right) }{\log s} \right\} \label{frm:upp and low bdsg} \\
     &=  \min_{s> 1}\left\{\frac{ \frac{((\sfA+R)s+\sfA)^2}2 + 2\sfA^2+ \log\left(2+\frac{(\sfA+R)s}{\sfA}\right)}{\log s} \right\} \text{,}\label{frm:opt cho for t}
  \end{align}
where \eqref{frm:zeros in real zeros in complex} follows because zeros of $f'_Y $ are also zeros of its complex extension $\breve f'_Y$;  \eqref{frm:lem: bd on the no of zeros of analytic func} is a consequence of Lemma~\ref{lem:number of zeros of analytic function}; \eqref{frm:upp and low bdsg} follows from Lemmas~\ref{lem:upper bound on maxima over a disk} and~\ref{lem:lower bound on maxima over a disk}; and finally, in \eqref{frm:opt cho for t}, we use the fact that $t= \frac {A}{R}$ is the minimizer in the right side of \eqref{frm:upp and low bdsg}. 
\end{proof}

Finally, combining the results of Lemmas~\ref{lem:LocationOFZeros_RealCase}~and~\ref{lem:extreme points and oscillations}, and Theorem~\ref{thm:NumberOfOscillations}, the following corollary presents the desired result of this section. 

\begin{corollary}\label{cor:bd on the zeros of f_Y - kappa_1}
Given an arbitrary constant $\kappa_1\in \left(0, \frac1{2\pi}\right)$, suppose $R>  \sfA + \log^{\frac{1}{2}} \left(\frac1{2 \pi \kappa_1^2}\right )$. Then, the number of zeros of $f_Y -\kappa_1$ satisfies
\begin{align}
&  \rmN(\bbR, f_Y -\kappa_1)  \notag\\
   &= \rmN([-R, R], f_Y -\kappa_1) \\ 
  &\le 1+  \min_{s> 1}\left\{\frac{ \left( \frac{((\sfA+R)s+\sfA)^2}2 + 2\sfA^2+ \log\left(2+\frac{(\sfA+R)s}{\sfA}\right) \right)}{\log s}\right\} \text{.} \label{eqn:bd on zeros of f_Y- kappa_1}
\end{align} 
\end{corollary}

\begin{rem}\label{rem:kappa's choice}
Observe that in presenting the main result of this section, a ``worst-case scenario" approach is taken. Indeed, the result in \eqref{eqn:bd on zeros of f_Y- kappa_1} is independent of the choice of $\kappa_1 $. If $\kappa_1 \approx 0 $, then $\rmN(\bbR, f_Y -\kappa_1) \le 2$ and the bound above may be quite loose. In applying Corollary~\ref{cor:bd on the zeros of f_Y - kappa_1} in the next section, we let 
\begin{align}
  \kappa_1 = \frac{\rme^{-C(\sfA)}}{\sqrt{2\pi\rme}} 
\end{align}
 where $C(\sfA) $ denotes the capacity of the amplitude constrained additive Gaussian channel. In that case, 
 it can be shown that 
 \begin{align}
 \left( 2\pi\rme \left(1+\sfA^2\right)\right)^{-\frac12} \le \kappa_1 \le \left(2\pi\rme +4\sfA^2\right)^{-\frac12} \text{,} \label{eq:BoundsONKappa}
 \end{align}
and the result presented above is more relevant.   
\end{rem}

\begin{rem}
 We  believe that the bound in Theorem~\ref{thm:NumberOfOscillations}, and hence the one in Corollary~\ref{cor:bd on the zeros of f_Y - kappa_1}, can be further tightened.  In fact, we conjecture that\footnote{ Let $f(x)$ and $g(x)$ be two nonnegative valued functions.   Then,  $f$ is  $\opTheta(g(x))$ if and only if   $c_1 g(x) \le f(x) \le c_2 g(x)$  for some $c_1,c_2>0$  and all $x>x_0$.  }
\begin{align}
\max_{X \in [-\sfA,\sfA]}\rmN(\mathbb{R},f_{Y}'  )= \opTheta(\sfA). \label{eq:ConjectureOnNumberOfZeros}
\end{align}
 Fig.~\ref{fig:NumberOfZeros} demonstrates a result of an extensive computer search that supports the claim in \eqref{eq:ConjectureOnNumberOfZeros} and compares it to the current best upper bound in Corollary~\ref{cor:bd on the zeros of f_Y - kappa_1}. 
  \begin{figure}[h!]
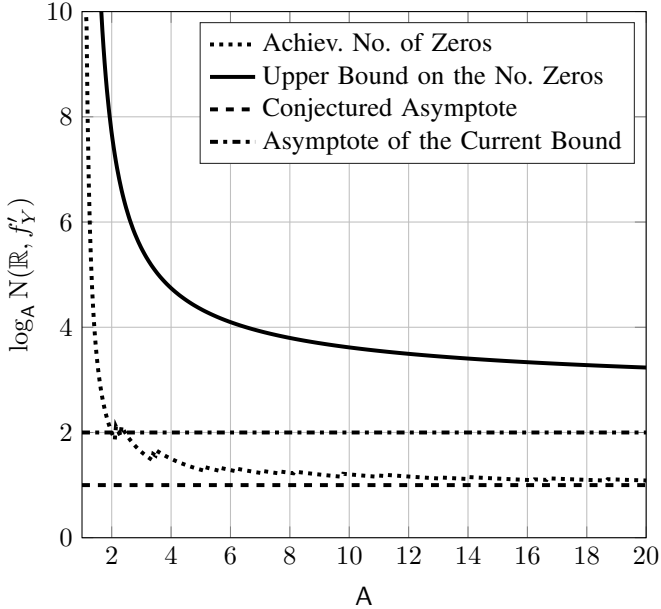

\center
%
%
%
\caption{Plot of the logarithm (in base $\sfA$) of the number of zeros of $f_Y'$. The solid black line uses the upper bound on the number of zeros in Corollary~\ref{cor:bd on the zeros of f_Y - kappa_1} with the bound on $\kappa_1$ in \eqref{eq:BoundsONKappa} and the dashed-dotted line is the asymptote of this upper bound.  The dashed line is the conjectured asymptote in \eqref{eq:ConjectureOnNumberOfZeros}.   The dotted line is the number of zeros found through a worst-case exhaustive numerical search.   }
\label{fig:NumberOfZeros}
\end{figure}

A possible bottleneck in our proof above is the bound in \eqref{frm:zeros in real zeros in complex}, where the function is extended to the complex plane, and the number of zeros are counted over a disk rather than over a finite interval. Doing so is effectively doubling the dimension of the problem. In other words, the produced order $\sfA^2$ bound follows because extending $f'_Y $ to the complex domain potentially creates a ton of zeros  that our analysis counts even though the function in the real domain cannot possibly realize those zeros. To work around this issue, one might consider using another version of Tijdeman's Lemma \cite[Lemma 1]{il1996counting} which works with arbitrary open sets, unlike the bound in Lemma~\ref{lem:number of zeros of analytic function} which works only over discs in the complex plane. Improvement is left for the future work. 
\end{rem}

\subsection{Proof of the Upper Bound in \eqref{eq:MainBound} } 

We begin by simplifying the previously provided upper bound on $\sfB_{\kappa_1}$. Note that an amplitude constraint $|X| \le \sfA$ induces a second moment constraint $\bbE[X^2]\le \sfA^2$, and therefore
\begin{align}
C(\sfA) &= \max_{\substack{ |X|\le \sfA\\ \bbE[X^2]\le \sfA^2 }} I(X;Y)     \\
 &\le \frac{1}{2} \log \left( 1+\sfA^2 \right)\text{.}  \label{akuaku}
\end{align} 
Since the differential entropy of a standard normal distribution is $h(Z) = \frac12 \log (2\pi \rme) $, \eqref{akuaku} implies that
\begin{align}
\frac1{\kappa_1} &= \exp(C(\sfA) + h(Z))  
\\ &\le  \sqrt{2 \pi \rme \left(1+\sfA^2\right)}\text{.} \label{eq:lowerBOundOnKappaALT}
\end{align} 
Capitalizing on the bound in \eqref{eq:BOundONTheNumberOfZeros},
\begin{align}
\sfB_{\kappa_1} &\le  \sfA+\sqrt{1+\log(1+\sfA^2)} \\
& \le 2\sfA+1 \text{,} \label{eqn:required later in the amp and pow const case}
	\end{align}
where the last inequality follows because  $\log(1+x) \le x$ and  $\sqrt{a+b} \le \sqrt{a}+\sqrt{b}$. 


As the finalizing step, letting $R \leftarrow (2\sfA + 1) $  in Theorem~\ref{thm:BoundingSupportWithOscillations} above, an application of Corollary~\ref{cor:bd on the zeros of f_Y - kappa_1} in Section~\ref{sec:BoundOnTheNumberOfOscillations} yields
\begin{align}
&\rmN([-\sfB_{\kappa_1},\sfB_{\kappa_1}], f_{Y^\star}  - \kappa_1 ) -1 \notag\\
&=\rmN([-2\sfA-1,2\sfA+1], f_{Y^\star}  - \kappa_1 )-1 \\
&\le  \min_{s> 1}\left\{\frac{\left( \frac{ \left(  (3\sfA+1)s+ \sfA \right)^2}2 + 2\sfA^2+ \log\left(2+\frac{( 3\sfA+1)s}{\sfA}\right) \right)}{\log s} \right\}\\
&\le  \min_{s> 1}\left\{\frac1{\log s} \left( \frac{ \left((3s+1)\sfA+ s \right)^2}2 + 2\sfA^2+ \log\left(2+4s\right) \right)\right\} \label{eq:BOundOnLogorithmiterms}\\
& \le  \left( \hspace{-0.03cm} \rme+ \hspace{-0.03cm} 2 \hspace{-0.04cm}\log \hspace{-0.04cm}\left(4 \sqrt{\rme}+2\right) + \hspace{-0.03cm} (4\rme+2\sqrt{\rme}) \sfA+ \hspace{-0.03cm}(5+ 4 \sqrt{\rme}+4 \rme)\sfA^2 \right) \label{eq:BoundWithS}\\
& = a_2 \sfA^2 +a_1 \sfA+ a_0 - 1 \text{,} \label{eq:BoundingLinearTerm}
\end{align}
where  \eqref{eq:BOundOnLogorithmiterms} follows because $3\sfA+1 \le 4\sfA$ for\footnote{The unessential assumption that $\sfA \ge 1 $ is just for simplifying the presentation. In the case when $\sfA\le 1 $, the optimality of $P_{X^\star}$ that is equiprobable on $\cX = \{-\sfA, \sfA \} $ is known  \cite{sharma2010transition}.} $\sfA \ge 1$;  \eqref{eq:BoundWithS} follows by choosing a suboptimal value $s= \sqrt{\rme}$ in the minimization; and  \eqref{eq:BoundingLinearTerm} follows by letting $a_2=9\rme+ 6\sqrt{\rme}+5 $,  $a_1=6\rme+2\sqrt{\rme}$ and  $a_0=\rme+2\log\left(4 \sqrt{\rme}+2\right)+1$.   

 \hfill \qedsymbol

\begin{rem} A more careful optimization of \eqref{eq:BoundWithS}  over the parameter $s$ would lead to better absolute constants $a_0 $, $a_1$ and $a_2$. However, note that the order $\sfA^2$  in \eqref{eq:BoundWithS} would not change.   
\end{rem}

\subsection{Proof of the Lower Bound in \eqref{eq:EPI_lowbd}}
Using the fact that  the optimizing input distribution is discrete with finitely many points and denoting by  $ H(P_{X^\star}) $ the entropy of the optimizing input distribution $P_{X^\star}$, it follows that
\begin{align}
\frac{1}{2} \log \left(1+\frac{2\sfA^2}{\pi \rme }  \right) &\le \max_{X \colon |X| \le \sfA} I(X;Y) \text{.}  \label{eq:ShannonLoweBound} \\
 &\le  H( P_{X^\star})  \label{entropy lanbutr} \\ 
 &\le \log \left( |\supp(P_{X^\star})|   \right) \text{,} \label{eq:EntropyBound}
\end{align}
where \eqref{eq:ShannonLoweBound} is a lower bound due to Shannon  \cite[Section 25]{Shannon:1948}.\hfill \qedsymbol

\section{Concluding Remarks}
\label{sec:Conclusion}

This paper has introduced several new tools for studying the capacity of the amplitude constrained additive Gaussian channels. Not only are the introduced tools strong enough to show that the optimal input distribution is discrete with finite support, but they are also able to provide concrete upper bounds on the number of elements in that support. The main result of this paper is that the number of zeros of the downward shifted optimal output density provides an implicit upper bound on the support size of the capacity-achieving input distribution. While this upper bound has been shown to be tight within a factor of one half, it can also be used as a means to get an explicit upper bound on the support size. 
  
 As a note on its flexibility, the novel method that is described in this paper has been demonstrated to be easily generalizable to other settings such as a scalar additive Gaussian channel with both peak and average power constraints.
 In addition to the scalar case, the method is shown to work for a vector Gaussian channel with an amplitude constraint $\sfA$.  In particular, for an optimal input  $\bsX^\star$, it has been shown that  its magnitude $\| \bsX^\star \|$ is  a discrete random variable with at most $ \rmO(\sfA^2)$ mass points for any fixed  dimension $n$.   

An interesting direction for further work in this area would be to sharpen the explicit bounds on the number of mass points.  Indeed, it has been conjectured with sufficient supporting arguments that the correct order on the number of points  should be $ \rmO(\sfA)$  rather than $ \rmO(\sfA^2)$. Although, finding a better explicit upper bound will ultimately still be related to finding the maximum number of oscillations of a Gaussian convolution within a bounded region.

As has been argued by Smith in  \cite[p.~40]{smith1969Thesis}, showing discreteness of the input distribution without providing bounds on the number of mass points does not reduce the maximization of the mutual information  from an infinite dimensional optimization (i.e., over the space of  all distributions) to the finite dimensional optimization (i.e., over $\mathbb{R}^n$ for some fixed $n$).    This issue has also been pointed out in \cite[p.~2346]{huang2005characterization}.  The results of this work, in fact, achieves this objective and reduce the infinite dimensional optimization over probability spaces to that in{\footnote{Considering the symmetries and properties of a distribution, the dimension of the search space can be reduced to $\bbR^{n}$.}} $\mathbb{R}^{2n}$ where $n=\rmO(\sfA^2)$, and  where $\mathbf{v} \in \mathbb{R}^{2n}$ consists $\mathbf{v}=[p_1,\ldots, p_n, x_1,\ldots, x_n]$ where $p_i$ is the probability mass of  the location $x_i$.   This dimensionality reduction potentially enables applications of efficient optimization algorithms with convergence guarantees such as the gradient descent and is the topic of our current investigation. 

It is highly likely that the presented approach generalizes to other (possibly non-additive) channels where channel transition probability is given by a strictly totally positive kernel (e.g., Poisson channel); the interested reader is referred to \cite{WCNC2019kernels} for a preliminary work on utilizing the techniques of this paper to non-additive settings.   The optimization technique used in this paper can also be adapted to other functionals over probability distributions such as the Bayesian minimum mean squared error; the interested reader is referred to \cite{ITW2018mmse} for this extension.

Finally, it would interesting to see if the results of this paper can be extended to  multiuser channels such as a multiple access channel with an amplitude constraint on the inputs where it is known that the discrete inputs are sum-capacity optimal \cite{mamandipoor2014capacity}, yet there are no bounds on the number of mass points of the optimal inputs.

\begin{appendices}
\section{Proof of Theorem~\ref{thm:vector Gaussian}}\label{app:thm:thm:vector Gaussian}

The starting point is the following sufficient and necessary conditions that can be found in\footnote{The most general result is shown in \cite{rassouli2016capacity}. However, a pleasing formulation such as the one in Lemma~\ref{lem:vector case necessary and sufficient cond.} is hidden behind the heavy notation of \cite{rassouli2016capacity}. We apply change of variables to provide much simpler presentation.} \cite{ShamQuadrat, rassouli2016capacity}. 
  
\begin{lemma}\label{lem:vector case necessary and sufficient cond.}
Consider the amplitude constrained vector additive Gaussian channel $\bsY = \bsX + \bsZ$ where the input $\bsX$, satisfying $\|\bsX \|\le \sfA $, is independent from the white Gaussian noise $\bsZ \sim \cN(\boldsymbol{0},\bfI_{n}) $. If $\bsX^\star$ is an optimal input, the distribution of its magnitude, namely  $P_{R^\star} = P_{\|\bsX^\star \|} $, satisfies
\begin{align}
i_{n}(r; P_{R^\star})  &  =   C_n(\sfA) + \nu_n \text{,} \quad   r \in \supp(P_{R^\star}) \text{,}\\
i_{n}(r; P_{R^\star})  & \le  C_n(\sfA) + \nu_n\text{,} \quad r \in [0, \sfA] \text{,} 
\end{align}
where $C_n(\sfA)$ denotes the capacity of the channel, and
\begin{align}
i_{n}(r; P_{R^\star}) &=  \int_0^\infty f_{\chi^2_n}(x|r)  \log  \frac{1}{g_n(x;P_{R^\star})} {\rmd}x \text{,} \\
f_{\chi^2_n}(x|r) &=  \frac12 \exp \left(- \frac{ x +r^2}{2} \right) \left(\frac{\sqrt{x}}{r}\right)^{\frac{n}{2}-1} \rmI_{\frac{n}{2}-1} \big(r \sqrt{x}\big)  \text{,} \label{eqn:def:noncentral chi squared}  \\
g_n(x;P_{R^\star})&= \int_0^A  \frac{2f_{\chi^2_n}(x|r)}{x^{\frac n2 -1} }  \rmd P_{R^\star}(r)\text{,} \label{eq:def:g_n(x;P_R)} \\
\nu_n &= \frac{n}{2} + \log \left(2^{\frac n2-1} \opGamma\left(\frac n2\right)\right) \label{eq:def:nu-VectorChannelCase}
 \text{,}
\end{align} 
with $\rmI_n(x)$ denoting the modified Bessel function of the first kind of order $n$. 
\end{lemma}
In a similar spirit to the proof of the scalar case, define
\begin{align}
\kappa_n &=  \exp(- C_n(\sfA) - \nu_n) \text{,} \label{eqn:def:kappa_n} \\
\Phi_n(s ; P_{R^\star}) &= i_{n}(s; P_{R^\star}) + \log \kappa_n  \text{,}  \label{eq:def:Phi_n}\\
\phi_n(x; P_{R^\star}) &=  \log \frac{\kappa_n }{g_n(x;P_{R^\star})}  \text{,} 
\end{align}
and observe that
\begin{align}
   \Phi_n(r ; P_{R^\star}) = \int_0^\infty  \phi_n(x; P_{R^\star}) f_{\chi^2_n}(x|r)  \rmd x \text{,}
\end{align}
where $f_{\chi^2_n}(x|r) $ is as defined in \eqref{eqn:def:noncentral chi squared}. Note that since $f_{\chi^2_n}(x|r) $ is the density of a non-central chi-squared distribution (with non-centrality parameter $r^2$, and degrees of freedom $n$), it is a strictly totally positive kernel \cite{karlin1956Polya1}. Hence, following the footprints of \eqref{frm:lemma 6}--\eqref{eq:ApplytingConcetrationBoundOFZeros},
\begin{align}
|\supp(P_{R^\star})| &\le \rmN \left([0, \sfA], \Phi_n(\cdot ; P_{R^\star})\right) \label{frm:lemma lagrangelike} \\
&\le 1+ \rmN \left((0, \infty) , \Phi_n(\cdot ; P_{R^\star})\right) \label{theExtra+1} \\
 &\le 1 + \rmN \left((0,\infty), \phi_n(\cdot ; P_{R^\star}) \right) \label{eq:OscillationTheoremComplex}\\
&= 1+ \rmN \left( (0,\infty) ,    g_n(\cdot ;P_{R^\star})  -  \kappa_n \right) \label{eq:SameNumberOFzerosAgainbhh}\\ 
&\le  1 +  \rmN \left( [0,\sfB_{\kappa_n}], g_n(\cdot ;P_{R^\star}) - \kappa_n \right) \label{eq:TrancationToBeta}\\ 
& \le  \mfra_{n_2} \sfA^2 + \mfra_{n_1} \sfA +\mfra_{n_0} \text{,} \label{eq:APpliyingTidjmans}
\end{align}  
where \eqref{frm:lemma lagrangelike} is a consequence of Lemma~\ref{lem:vector case necessary and sufficient cond.}; the extra $+1$ in \eqref{theExtra+1} is just to account for the possibility that $\Phi_n(0 ; P_{R^\star}) = 0 $; \eqref{eq:OscillationTheoremComplex} follows from Karlin's Oscillation Theorem, see Theorem~\ref{thm:OscillationThoerem}; \eqref{eq:SameNumberOFzerosAgainbhh} follows since $\phi_n(\cdot ; P_{R^\star})$ has the same zeros as  $ g_n(\cdot ;P_{R^\star})  - \kappa_n $; \eqref{eq:TrancationToBeta} follows from Lemma~\ref{lem:LocationOFZeros_VectorCase} in Appendix~\ref{apdx:Additional Lemmas Vector Case};  and \eqref{eq:APpliyingTidjmans} is shown in Lemma~\ref{lem:BoundOnNUmberofZerosComplex} that can be found in Appendix~\ref{apdx:Additional Lemmas Vector Case}.

%

\section{Additional Lemmas for the Upper Bound Proof of Theorem~\ref{thm:vector Gaussian}} \label{apdx:Additional Lemmas Vector Case}
This section contains several supplementary lemmas that are used in the upper bound proof of Theorem~\ref{thm:vector Gaussian}.
\begin{lemma}\label{lem:ExponBoundBessel} For $n\in \bbN$ and $z\in \bbC$
\begin{align}
| \rmI_n(z)| \le \frac{\sqrt{\pi} |z|^n}{2^n\opGamma(n+\frac12)} \rme^{\left|{\mathfrak Re} (z)\right|} \text{.} 
\end{align}	
	
\end{lemma}
\begin{proof}
Thanks to the integral representation of the modified Bessel function of the first kind, see \cite[9.6.18]{abramowitz1966handbook},
\begin{align}
\rmI_n(z)=\frac{(\frac12 z)^n}{\sqrt{\pi}\opGamma(n+\frac12 ) } \int_0^\pi \rme^{z \cos(\theta)} \sin^{2n} (\theta) {\rm d} \theta \text{,}	
\end{align}
it follows from  the modulus inequality that
\begin{align}
	  | \rmI_n(z)| & \le \frac{(\frac12|z|)^n}{\sqrt{\pi}\opGamma(n+\frac12)} \int_0^\pi |\rme^{ z \cos(\theta)}||\sin^{2n}( \theta)| {\rm d} \theta \\
	 	  &\le \frac{\sqrt{\pi} |z|^n}{2^n\opGamma(n+\frac12)} \rme^{\left|{\mathfrak Re} (z)\right| } \text{,} \label{simple1eqeq}
\end{align}
	where \eqref{simple1eqeq} follows because $|\rme^{ z \cos(\theta)}||\sin^{2n}( \theta)| \le \rme^{\left|{\mathfrak Re} (z)\right|}$.
\end{proof}
Similar to its counterpart in Lemma~\ref{lem:LocationOFZeros_RealCase}, the next lemma provides a bound  on the interval for zeros of the function $g_n(\cdot;P_{R^\star}) - \kappa_n $.    
\begin{lemma}[On the Location and Finiteness of Zeros of $g_n(\cdot ;P_R)-\kappa_n$]\label{lem:LocationOFZeros_VectorCase} Given an arbitrary distribution $P_R $, for a fixed $\kappa_n \in \left(0, 1 \right]$ there exists some $\sfB_{\kappa_n}<\infty$ such that 
\begin{align}
 \rmN \left([0, \infty), g_n(\cdot ;P_R)-\kappa_n \right) &=\rmN \left(\left[0, \sfB_{\kappa_n} \right], g_n(\cdot;P_R)-\kappa_n \right) \\
 &<\infty\text{.}
 \end{align} 
In particular, there are finitely many zeros of $g_n(\cdot ;P_R)-\kappa_n$ all of which are contained within the interval $ [0, \sfB_{\kappa_n}]$. Moreover, $\sfB_{\kappa_n}$ can be upper bounded as follows: 
\begin{align}
\sfB_{\kappa_n} \le \left(\sfA + \sqrt{\sfA^2 + 2 \log \left(\frac{\gamma_n}{\kappa_n}\right)}\right)^2  \text{,}  \label{eq:BoundOnBkappaComplexCase}
\end{align}
where 
\begin{align}
\gamma_n =\frac{\sqrt{\pi} }{2^{\frac{n}{2} -1}\opGamma(\frac{n-1}{2})}\text{.} \label{eqn:def:gamma_n}
\end{align} 
\end{lemma} 
\begin{proof} 
From the definition of the pdf $ g_n(\cdot ;P_R)$ in \eqref{eq:def:g_n(x;P_R)},
\begin{align}
&g_n(x ;P_R) \notag\\
 &= \int_0^\sfA    \exp \left(- \frac{ x +r^2}{2} \right)  \frac{ \rmI_{\frac{n}{2}-1} \big( r \sqrt{x}  \big)}{ \left( r\sqrt{x} \right)^{\frac{n}{2}-1}} {\rm d}P_{R}(r)\\
&\le \int_0^\sfA  \hspace{-0.03cm}  \frac{\sqrt{\pi}}{2^{\frac{n}{2}-1}\opGamma(\frac{n-1}{2})} \exp \left(- \frac 12\left(\sqrt{x}-r\right)^2 \right) {\rm d}P_{R}(r) \label{eq:BoundOnBesselFunction}\\
&\le \frac{\sqrt{\pi} }{2^{\frac{n}{2}-1}\opGamma(\frac{n-1}{2})}   \exp \left(- \frac{x}{2} +  \sfA\sqrt{x} \right)  \text{,}\label{eq:BoundsInExponents}
\end{align}
where  \eqref{eq:BoundOnBesselFunction} follows from Lemma~\ref{lem:ExponBoundBessel}; and \eqref{eq:BoundsInExponents} utilizes $r\in [0, \sfA]$. Since the right side of \eqref{eq:BoundsInExponents} is decreasing for all $x> \sfA^2$, it follows that 
\begin{align}
 g_n(x; P_R)-\kappa_n < 0 
\end{align}
for all
\begin{align}
	x >  \left(\sfA + \sqrt{\sfA^2 + 2 \log \left(\frac{\gamma_n}{\kappa_n}\right)}\right)^2 \text{.}
\end{align}
This means that there exists a $\sfB_{\kappa_n} $ satisfying \eqref{eq:BoundOnBkappaComplexCase} such that all zeros of $ g_n(\cdot; P_R)-\kappa_n $ are contained within the interval $[0, \sfB_{\kappa_n}] $.

To see that there are finitely many zeros of $ g_n(\cdot; P_R)-\kappa_n $, using the fact that $ g_n(\cdot; P_R)$ is analytic\footnote{For a proof, refer to \cite[Appendix I]{ShamQuadrat} and \cite[Propositions~1~and~2]{rassouli2016capacity} for the respective cases of $n=2$, and $n\ge 2$.} suffices, because non-zero analytic functions can only have finitely many zeros on a compact interval.
\end{proof}
Following the footsteps of the upper bound proof in the scalar case, the evaluation of the derivative of the function $g_n(\cdot; P_R)$ is established next. 
\begin{lemma}\label{lem:derivative of g_n}
For $x\in (0, \infty)$
	\begin{align}
		&\frac{\rmd }{\rmd x} g_n(x;P_R) \notag\\
		& = \bbE\left[ \frac{\exp\left(-\frac{x+R^2}{2}\right)}{2(R\sqrt{x})^{\frac n2 -1}} \left(\frac{R}{\sqrt{x}} \rmI_{\frac n2}(R\sqrt{x}) - \rmI_{\frac n2 -1}(R\sqrt{x})\right)\right]  \text{,}
	\end{align}
	where $R\sim P_R $.
\end{lemma}

\begin{proof}
First of all, given $r>0 $, observe that for $u\in (0, \infty)$ 
\begin{align}
 t_n(u|r) &=	2\left(\frac{r}{u}\right)^{n -2} \exp \left(\frac {r^2}{2} \right)  f_{\chi^2_n}\left(\frac{u^2}{r^2} \middle|r\right)  \\ &= \rme^{-\frac{u^2}{2r^2}}  \frac{\rmI_{\frac n2 -1} (u )}{u^{\frac n2-1}}
\end{align}
is a differentiable function and
\begin{align}
	\frac{\rmd}{\rmd u}  t_n(u|r) = \frac{\rme^{-\frac{u^2}{2r^2}}}{u^{\frac n2-1}} \left(\rmI_{\frac n2}(u) - \frac{u}{r^2} \rmI_{\frac n2 -1}(u) \right) \text{,}
\end{align}
where we have employed the fact that \cite[Eqn.~(9.6.26)]{abramowitz1966handbook}
\begin{align}
	\frac{\rmd}{\rmd u}\rmI_{\frac n2 -1}(u) = \rmI_{\frac n2}(u) + \frac{n-2}{2u} \rmI_{\frac n2-1}(u) \text{.}
\end{align}
The desired result then follows from the chain rule as 
\begin{align}
	g_n(x;P_R) = \int_0^\sfA  \rme^{-r^2/2} t_n(r\sqrt{x}|r) \rmd P_R(r) \text{.}
\end{align}
\end{proof}
As was the case in the scalar Gaussian channel, we shall analyze the complex extension of the derivative of $g_n(x;P_R)$. For this reason, in what follows, we denote the complex extension of the derivative of $g_n(x;P_R)$ by $\breve g'_n(x;P_R)$.

\begin{lemma}\label{lem:DiffOfBessel}  Given $r>0$ and $\sfD>0$
  \begin{align}
& \rmI_{\frac n2 -1}(\sfD r)  - \frac{r}{\sfD} \rmI_{\frac n2}(\sfD r) \notag\\
 &\ge  (\sfD r)^{\frac n2 -1}  \frac{2^{1-\frac n2}}{\opGamma( \frac{n}{2} )}  \left(1- \frac{2r^2}{n-1+\sqrt{(n-1)^2+(2\sfD r)^2}} \right)  \\
 &>0 \text{.}
  \end{align}
\end{lemma}
\begin{proof}
Using the fact that $I_n(x)>0$ for $x>0$
  \begin{align}
  &  \rmI_{\frac n2 -1}(\sfD r) \left( 1  -\frac{r}{\sfD} \frac{ \rmI_{\frac n2}(\sfD r)}{\rmI_{\frac n2 -1}(\sfD r)}\right)  \notag\\
    & \ge    \rmI_{\frac n2 -1}  \left( \sfD r \right)  \left(1- \frac{2r^2}{n-1+\sqrt{(n-1)^2+(2\sfD r)^2}} \right)  \label{eq:BoundsOnRatioBessel} \\
    & \ge   (\sfD r)^{\frac n2 -1}  \frac{2^{1-\frac n2}}{\opGamma( \frac{n}{2} )}   \left(1- \frac{2r^2}{n-1+\sqrt{(n-1)^2+(2\sfD r)^2}} \right) \text{,} \label{eq:Monotonicity}
  \end{align}
where \eqref{eq:BoundsOnRatioBessel}  follows from (see \cite[Theorem~1]{segura2011bounds})
\begin{align}
\frac{\rmI_{\frac{n}{2}}(x)}{\rmI_{\frac{n}{2}-1}(x)} \le  \frac{2x}{ n-1+\sqrt{(n-1)^2+(2x)^2}} \text{;}  
\end{align} 
and \eqref{eq:Monotonicity} follows from the fact that $x^{-n} \sfI_n( x) $ is monotonically increasing for $x>0$ and that\footnote{See \cite[Eqn.~(9.6.28)]{abramowitz1966handbook}, and \cite[Eqn.~(9.6.7)]{abramowitz1966handbook}, respectively.} 
\begin{align}
  \lim_{x \to 0} x^{-n} \sfI_n( x) = 2^{-n} \opGamma^{-1} \left(n+1\right) \text{.}
\end{align}    
\end{proof}
To be plugged into the Tijdeman's Number of Zeros Lemma, Lemmas~\ref{lem:LOwerBoundONDerivativeOfPdf}~and~\ref{lem:UpperBoundONDerivativeOfPdf} find useful suboptimal lower and upper bounds for the maximum value of $\breve g_n' \left( \cdot ;P_R \right)$ on a disc centered at $z_0 \in \bbC$ where
\begin{align}
 z_0 =  \frac{\sfB_{\kappa_n}}{2} + i 0 \text{.}
\end{align} 
\begin{lemma}\label{lem:LOwerBoundONDerivativeOfPdf} Suppose $\sfD>0$. For $ \sfB_{\kappa_n} \le 2\sfD^2 $
\begin{align}
 &  \max_{ |z| \le \sfD^2}  \left|\breve g_n' \left(z+\frac{\sfB_{\kappa_n}}{2};P_R \right) \right|  2^{\frac n2}  \opGamma \left( \frac{n}{2} \right)\exp\left(\frac{\sfD^2+\sfA^2}{2}\right)  \notag\\
 &\qquad \ge   \left(1- \frac{2\sfA^2}{n-1+\sqrt{(n-1)^2+(2\sfD \sfA)^2}} \right) \text{.}
\end{align}
\end{lemma} 
\begin{proof}
Observe that for $R\sim P_R$
\begin{align}
  & \max_{ |z| \le \sfD^2}  \left|\breve g_n' \left(z+\frac{B_{\kappa_n}}{2};P_R \right) \right|  \ge \left|\breve g_n' \left(\sfD^2 ;P_R \right) \right|  \label{eq: choosingsuboptimalZ}\\
		&\quad =   \left|\bbE\left[ \frac{\exp\left(-\frac{\sfD^2+R^2}{2}\right)}{2(\sfD R)^{\frac n2 -1}} \left(\frac{R}{\sfD} \rmI_{\frac n2}(\sfD R) - \rmI_{\frac n2 -1}(\sfD R)\right)\right]   \right|  \\
		&\quad \ge \bbE \left[ \exp\left(-\frac{\sfD^2+R^2}{2}\right) \frac{2^{-\frac n2}}{\opGamma( \frac{n}{2} )} \right. \notag \\
& \qquad \left. \cdot 		  \left(1- \frac{2R^2}{n-1+\sqrt{(n-1)^2+(2\sfD R)^2}} \right)  \right] \label{eq:mon_offuncinside} \\
   		&\quad \ge  \frac{2^{-\frac n2}}{\opGamma( \frac{n}{2} )}  \exp\left(-\frac{\sfD^2+\sfA^2}{2}\right)  \notag\\
		&\qquad \cdot \left(1- \frac{2\sfA^2}{n-1+\sqrt{(n-1)^2+(2\sfD \sfA)^2}} \right) \text{,}  \label{eq:monotonictyOfExponential}
\end{align}
where \eqref{eq: choosingsuboptimalZ} follows by choosing a suboptimal value of $z= \sfD^2 - \frac{\sfB_{\kappa_n}}{2} $; \eqref{eq:mon_offuncinside} follows from Lemma~\ref{lem:DiffOfBessel}; and  \eqref{eq:monotonictyOfExponential} follows because $R\le \sfA $.
\end{proof}

\begin{lemma}\label{lem:UpperBoundONDerivativeOfPdf}  Suppose $\sfM>0$. For $ \sfB_{\kappa_n} \le 2\sfM^2$
\begin{align}
	&\max_{ |z| \le \sfM^2 }  \left|\breve g_n' \left(z+\frac{\sfB_{\kappa_n}}{2};P_R \right)\right| \notag\\
	& \le   \frac{\gamma_n}{2}  \left(\frac{\sfA^2}{n-1}+1\right)  \exp \left(\frac{1}{2}\left(\sfA+\sqrt{2}\sfM\right)^2\right) 
  \text{,}
\end{align}		
where $\gamma_n $ is as defined in \eqref{eqn:def:gamma_n}.
\end{lemma}
\begin{proof}
Capitalizing on the result of Lemma~\ref{lem:derivative of g_n}, the complex extension of the derivative of $g_n \left(x;P_R \right) $ satisfies
  \begin{align}
  &   \left|\breve g_n' \left(z;P_R \right) \right| \notag\\
     &= \left| \bbE\left[ \frac{\exp\left(-\frac{z+R^2}{2}\right)}{2(R\sqrt{z})^{\frac n2 -1}} \left(\frac{R}{\sqrt{z}} \rmI_{\frac n2}(R\sqrt{z}) - \rmI_{\frac n2 -1}(R\sqrt{z})\right)\right] \right| \\
     &\le  \bbE \hspace{-0.05cm}\left[\left| \frac{\exp\left(-\frac{z+R^2}{2}\right)}{2(R\sqrt{z})^{\frac n2 -1}}\right|   \hspace{-0.05cm} \left(\left| \frac{R}{\sqrt{z}} \rmI_{\frac n2}(R\sqrt{z}) \right|  \hspace{-0.05cm} +   \hspace{-0.05cm} \left| \rmI_{\frac n2 -1}(R\sqrt{z})\right|\right)   \hspace{-0.05cm} \right] \label{eq:ModulusANDtriangularInequality}  \\
     &\le  \bbE\left[ \frac{\sqrt{\pi} } {2^{\frac n2}\opGamma(\frac{n-1}{2})} \left(\frac{R^2}{n-1}+1\right)  \rme^{{\mathfrak Re}\left(-\frac{(R \pm \sqrt{z})^2}{2}\right) } \right], \label{frm:lemmtooaboce}     
  \end{align}
  where \eqref{eq:ModulusANDtriangularInequality} follows from subsequent applications of modulus and triangular inequalities; \eqref{frm:lemmtooaboce} is a consequence of Lemma~\ref{lem:ExponBoundBessel}. 
To finalize the proof, using the fact that $R\in [0, \sfA]$, we simply observe that
  \begin{align}
  &\max_{ |z| \le \sfM^2}   \left|\breve g_n' \left(z+\frac{\sfB_{\kappa_n}}{2} ;P_R \right)\right| \notag\\
  &\qquad \le \max_{ |z| \le \sfM^2 }  \bbE\left[ \frac{\sqrt{\pi}  \left(\frac{R^2}{n-1}+1\right)  } {2^{\frac n2}\opGamma(\frac{n-1}{2})}  \rme^{ {\mathfrak Re}\left(-\frac{1}{2}\left(R \pm \sqrt{z+\frac{\sfB_{\kappa_n}}{2}}\right)^2\right) } \right] \\
  &\qquad \le \max_{ |z| \le \sfM^2}  \bbE\left[ \frac{\sqrt{\pi} \left(\frac{R^2}{n-1}+1\right) } {2^{\frac n2}\opGamma(\frac{n-1}{2})}   \rme^{\frac12 \left(|R|+ \left(|z|+\frac{\sfB_{\kappa_n}}{2} \right)^{\frac12}\right) ^2 } \right]\label{frm:abs val is gr and triang} \\
  &\qquad \le  \frac{\sqrt{\pi} } {2^{\frac n2}\opGamma(\frac{n-1}{2})} \left(\frac{\sfA^2}{n-1}+1\right)  \exp \left(\frac{1}{2}\left(\sfA+ \sqrt{2}\sfM\right)^2\right) \text{,} \label{frm:trian222}
  \end{align}
where \eqref{frm:abs val is gr and triang} follows after realizing ${\mathfrak Re}(z) \le |z|  $, and applying the triangle inequality twice.
\end{proof}
Assembling the results of Lemmas~\ref{lem:LocationOFZeros_VectorCase},~\ref{lem:LOwerBoundONDerivativeOfPdf},~and~\ref{lem:UpperBoundONDerivativeOfPdf}, togerher with Tijdeman's Number of Zeros Lemma, i.e., Lemma~\ref{lem:number of zeros of analytic function}, the following result establishes a suboptimal upper bound on the number of zeros of the function $g_n(\cdot ;P_R) - \kappa_n$. 
\begin{lemma}\label{lem:BoundOnNUmberofZerosComplex}
Suppose that $\supp(P_R)\in [0, \sfA] $ and $\sfB_{\kappa_n} $ is as defined in Lemma~\ref{lem:LocationOFZeros_VectorCase}. The number of zeros of $g_n(\cdot ;P_R) - \kappa_n $ within $[0, \sfB_{\kappa_n}]$ satisfies 
\begin{align}
\rmN \left([0, \sfB_{\kappa_n}], g_n(\cdot ;P_R) - \kappa_n \right) \le  \mfra_{n_2} \sfA^2 + \mfra_{n_1} \sfA +\mfra_{n_0}-1 \text{,}
\end{align}	
where $\mfra_{n_2}$, $\mfra_{n_1}$, and $\mfra_{n_0}$ are as defined in \eqref{eqn:def:a_n_2Vectorcase}, \eqref{eqn:def:a_n_1Vectorcase}, and \eqref{eqn:def:a_n_0Vectorcase}, respectively
\end{lemma}
\begin{proof}
In light of Lemma~\ref{lem:LocationOFZeros_VectorCase}, let
\begin{align}
\overline \sfB_{\kappa_n} = \left(\sfA + \sqrt{\sfA^2 + 2 \log \left(\frac{\gamma_n}{\kappa_n}\right)}\right)^2 \text{,} 
\end{align}
and note that
\begin{align}
  &\rmN\left([0, \sfB_{\kappa_n}], g_n(\cdot ;P_R) - \kappa_n\right) \notag \\
  &\le 1 + \rmN\left([0, \sfB_{\kappa_n}], g'_n(\cdot ;P_R) \right) \label{frm:Rolle roll eyes} \\
  &= 1 + \rmN\left(\left[ -\frac{\sfB_{\kappa_n}}{2},  \frac{\sfB_{\kappa_n}}{2}\right], g'_n\left(\cdot + \frac{\sfB_{\kappa_n}}{2} ;P_R\right) \right)  \\
  &\le 1+ \rmN\left(\cD_{\frac{\sfB_{\kappa_n}}{2}},\breve g'_n\left(\cdot + \frac{\sfB_{\kappa_n}}{2} ;P_R\right) \right)\label{frm:extend complex} \\
  &\le 1 + \min_{s>1, \, t>0} \left\{ \frac{\log \frac{   \max_{|2z|\le (st+s+t)\sfB_{\kappa_n} }  \left|\breve g'_n\left(z + \frac{\sfB_{\kappa_n}}{2} ;P_R\right)   \right|  }{ \max_{|2z|\le t\sfB_{\kappa_n} }  \left|\breve g'_n\left(z + \frac{\sfB_{\kappa_n}}{2} ;P_R\right)   \right|  }  }{\log s}  \right\} \label{frm:tijd22} \\
  &\le 1 +   \max_{|2z|\le (2\rme + 1)\overline \sfB_{\kappa_n} } \log \left|\breve g'_n\left(z + \frac{\sfB_{\kappa_n}}{2} ;P_R\right) \right|  \notag\\
  &  \qquad  - \max_{|2z|\le \overline \sfB_{\kappa_n}  } \log \left|\breve g'_n\left(z + \frac{\sfB_{\kappa_n}}{2} ;P_R\right) \right|    \label{frm:s=e,t=M/2B-1} \\
  &\le \log \frac{\rme \sqrt{\pi}\opGamma\left(\frac n2\right) }{\opGamma\left(\frac{n-1}{2}\right)} + \left(\frac{3}{4}+\rme\right) \overline \sfB_{\kappa_n} + \sfA^2 + \sqrt{2\rme+1} \sfA \overline \sfB_{\kappa_n}^\frac12 \notag\\
  &   \qquad + \log \left( \frac{\sfA^2}{n-1} +1\right)  \notag \\ 
  & \qquad - \log\left(1- \frac{2\sfA^2}{n-1 + \sqrt{(n-1)^2 + 2\overline \sfB_{\kappa_n} \sfA^2 } }\right)  \label{frm:sub opt choi alx} \\
  &\le \log \frac{\rme \sqrt{\pi}\opGamma\left(\frac n2\right) }{\opGamma\left(\frac{n-1}{2}\right)} + \left(\frac{3}{4}+\rme\right) \overline \sfB_{\kappa_n} + \sfA^2 + \sqrt{2\rme+1} \sfA \overline \sfB_{\kappa_n}^\frac12 \notag\\
  & + \sqrt{\frac{32}{n-1}} \sfA\label{frm:tedious algebra, thanks alex!} \\
  &\le  \mfra_{n_2} \sfA^2  + \mfra_{n_1} \sfA  + \mfra_{n_0} - 1 \text{,}   \label{frm:final looseness}
\end{align}
where \eqref{frm:Rolle roll eyes} follows from Rolle's Theorem; in \eqref{frm:extend complex} $\cD_r\subset \bbC$ denotes a disk of radius $r$ centered at the origin and the bound follows because zeros of $g'_n$ are also zeros of its complex extension $\breve g'_n$; \eqref{frm:tijd22} follows from Tijdeman's Number of Zeros Lemma, see Lemma~\ref{lem:number of zeros of analytic function}; \eqref{frm:s=e,t=M/2B-1} follows from the suboptimal choices: 
\begin{align}
  s &= \rme \text{,} \\ 
  t &= \frac{\overline \sfB_{\kappa_n}}{\sfB_{\kappa_n}}  \text{;}
\end{align}
\eqref{frm:sub opt choi alx} follows from Lemmas~\ref{lem:LOwerBoundONDerivativeOfPdf}~and~\ref{lem:UpperBoundONDerivativeOfPdf} with
\begin{align}
  \sfD^2 &\leftarrow \frac{1}{2} \overline \sfB_{\kappa_n} \text{,}  \\
  \sfM^2 &\leftarrow \frac{2\rme +1}{2} \overline \sfB_{\kappa_n} \text{;} 
\end{align} 
\eqref{frm:tedious algebra, thanks alex!} follows  from a tedious algebra where we first note, from their definitions in \eqref{eqn:def:kappa_n} and \eqref{eqn:def:gamma_n}, that the ratio $\gamma_n/\kappa_n >1 $, implying $\overline \sfB_{\kappa_n} > 2\sfA^2 $, implying
\begin{align}
  1- \frac{2\sfA^2}{n-1 + \sqrt{(n-1)^2 + 2\overline \sfB_{\kappa_n} \sfA^2}}\ge\left(\frac{2\sfA^2}{n-1}+1\right)^{-1} \text{,}
\end{align}
and allowing us to upper bound the last two ``$\log$" terms in the right side of \eqref{frm:sub opt choi alx} by 
\begin{align}
  2\log\left(\frac{2\sfA^2}{n-1} +1\right)  \le \left(\frac{32}{n-1}\right)^\frac12 \sfA \text{;}
\end{align}
finally \eqref{frm:final looseness} follows from the definitions of $\kappa_n$, $\gamma_n$ (in \eqref{eqn:def:kappa_n}, and \eqref{eqn:def:gamma_n}, respectively) and the facts that
\begin{align}
 \sfA \sqrt{\sfA^2 + 2 \log \left(\frac{\gamma_n}{\kappa_n}\right)} &\le \sfA^2+ \log \left(\frac{\gamma_n}{\kappa_n}\right) \text{,} \\
 C_n(\sfA) &\le \frac{n}{2} \log(1+\sfA^2) \le n\sfA \text{.}
\end{align}
\end{proof}
\section{Proof of Theorem~\ref{thm:AGC amplitude and power constraint}}\label{apdx:AGC amplitude and power constraint}
\subsection{Proof of the Upper Bound in Theorem~\ref{thm:AGC amplitude and power constraint}}
The first ingredient of the upper bound proof is once again due to Smith \cite[Corollary 2]{smith1971information} who characterizes the optimal input distribution as follows. 
\begin{lemma}\label{lem:Smith's Lemma on Variance Constraint Problem}
Consider the amplitude and power constrained scalar additive Gaussian channel $Y=X+Z$ where the input $X$, satisfying $|X|\le \sfA$ and $\bbE[|X|^2]\le \sfP $, is independent from the noise $Z\sim \cN(0,1) $. Then,
 $P_{X^\star} $ is the capacity-achieving input distribution if and only if the following conditions are satisfied:
\begin{align}
  i(x; P_{X^\star}) &= C(\sfA, \sfP) + \lambda (x^2 - \sfP) , \quad x\in \supp(P_{X^\star}) \text{,} \\
  i(x; P_{X^\star}) &\le C(\sfA, \sfP) + \lambda (x^2 - \sfP), \quad x\in [-\sfA, \sfA] \text{,} \\
  0  &=  \lambda(\sfP- \bbE[X^2] ) \text{,}
\end{align}
where $C(\sfA, \sfP)$ denotes the capacity of the channel, and $i(x; P_{X^\star})$ is as defined in \eqref{eqn:lem:def:i(x;P_X^star)}. 
\end{lemma}
\begin{rem}
Hidden in our notation for typographic reasons, the Lagrange multiplier $\lambda $ in fact depends on amplitude and power constraints, namely $\sfA$ and $\sfP$. Indeed, since $|X| \le \sfA$, if $\sfP>\sfA^2$, the power constraint is inactive, implying $\lambda = 0$. In this case, the problem reduces to additive Gaussian channel with only amplitude constraint, and we recover Lemma~\ref{lem:SmithResult}. 
\end{rem}

As a corollary to above lemma, note that if $x $ is a point of support of $P_{X^\star} $ (i.e., $x\in \supp(P_{X^\star})$), then $x$ is a zero of the function 
\begin{align}
  \Xi_{\sfA,\, \sfP}(x;P_{X^\star}) = i(x;P_{X^\star} ) -  C(\sfA, \sfP ) - \lambda (x^2 - \sfP) \text{.}
\end{align}
In other words, 
\begin{align}
  |\supp(P_{X^\star}) | &\le \rmN([-\sfA, \sfA], \Xi_{\sfA,\, \sfP}(\cdot; P_{X^\star})) \\
  &\le \rmN(\bbR, \Xi_{\sfA,\, \sfP}(\cdot; P_{X^\star})) \text{.} \label{eqn:initialize in power and amp case}
\end{align}
Observe that,  since
\begin{align}
  x^2  = \int_\bbR \frac{y^2-1}{\sqrt{2\pi}} \rme^{-\frac{(y-x)^2}{2}} \rmd y \text{,}
\end{align}
we can write
\begin{align}
   \Xi_{\sfA,\, \sfP}(x;P_{X^\star}) = \int_{\bbR}  \frac{\xi_{\sfA,\, \sfP}(y)}{\sqrt{2\pi}} \rme^{-\frac{(y-x)^2}{2}} \rmd y \text{,}
\end{align}
where
\begin{align}
 \xi_{\sfA,\, \sfP}(y) =  \log\frac1{f_{Y^\star}(y)} -h(Z) - C(\sfA, \sfP) + \lambda \sfP  - \lambda (y^2 -1) \text{.} \label{eq:de:xi_AP}
\end{align}
Keeping the steps \eqref{frm:lemma 6}--\eqref{eq:ApplytingConcetrationBoundOFZeros} in mind, define 
\begin{align}
\mfrF_{\sfA,\, \sfP}^\star (y)  &= \rme^{\lambda y^2} f_{Y^\star}(y) - \kappa_{\sfA,\, \sfP} \text{,} \label{eq:de:mfrF_{A,P}} 
\end{align}
with\footnote{Note that $0\le i(0; P_{X^\star}) \le C(\sfA,\sfP)-\lambda\sfP$, and $\lambda < 1/2 $, cf. Lemma~\ref{lem:lambda value}. This implies that $\kappa_{\sfA, \, \sfP} < 1/\sqrt{2\pi} $.}
\begin{align}
  \kappa_{\sfA,\, \sfP}=\exp(-h(Z)-C(\sfA,\sfP)+\lambda(\sfP +1)) \text{.} \label{eq:def:kappa_{A,P}}
\end{align}
Using the fact that the Gaussian distribution is a strictly totally positive kernel, and resuming from \eqref{eqn:initialize in power and amp case}
\begin{align}
 |\supp(P_{X^\star}) | &\le  \rmN\left(\bbR,  \Xi_{\sfA,\, \sfP}(\cdot;P_{X^\star})\right) \\ &\le \rmN\left(\bbR, \xi_{\sfA,\, \sfP}\right) \label{frm:Karlin osc pow} \\
  &= \rmN\left(\bbR, \mfrF_{\sfA,\, \sfP}^\star\right) \label{eqn:same no of zeros} \\
  &= \rmN\left([-\sfB_{\kappa_{\sfA,\, \sfP}}, \sfB_{\kappa_{\sfA,\, \sfP}}], \mfrF_{\sfA,\, \sfP}^\star\right) \label{frm:loca of zero power} \\
  &\le  a^{}_{\sfP_2} \sfA_\sfP^2 + a^{}_{\sfP_1} \sfA_\sfP + a^{}_{\sfP_0} \text{,} \label{frm:lem:bound using tijdeman in power case}
\end{align}
where \eqref{frm:Karlin osc pow} follows from Karlin's Oscillation Theorem, see Theorem~\ref{thm:OscillationThoerem}; \eqref{eqn:same no of zeros} follows because $\xi_{\sfA,\, \sfP}(y)= 0$ if and only if $ \mfrF_{\sfA,\, \sfP}^\star (y) = 0$; \eqref{frm:loca of zero power} is a consequence Lemma~\ref{lem:location of zeros power constraint} in Appendix~\ref{apdx:AGC amplitude-power supplement}; finally \eqref{frm:lem:bound using tijdeman in power case} follows from Lemma~\ref{lem:bound using tijdeman power case}  in Appendix~\ref{apdx:AGC amplitude-power supplement}.
\qed 

\subsection{Proof of the Lower Bound in Theorem~\ref{thm:AGC amplitude and power constraint}}

Invoking entropy-power inequality,
\begin{align}
I(X;Y)&=  h(X+Z)-h(Z)\\
& \ge \frac12 \log \left( \rme^{2h(X)} +\rme^{2h(Z)} \right) -h(Z)\\
&= \frac12 \log \left(  \frac{1}{2 \pi \rme } \rme^{2h(X)} + 1 \right) \text{.}
\end{align}
Therefore,
\begin{align}
\log \left|\supp(P_{X^\star} )\right| &\ge H(P_{X^\star})  \\ 
&\ge  \max_{X:  |X| \le \sfA,\, \bbE[X^2] \le \sfP} I(X;Y) \\ 
&\ge  \max_{X:  |X| \le \sfA,\, \bbE[X^2] \le \sfP} \frac12 \log \left(  \frac{\rme^{2h(X)}}{2 \pi \rme} + 1 \right)\\
& \ge \max_{  |a| \le \sfA,\,  \frac{a^2}{3} \le \sfP} \frac12 \log \left(  \frac{2 a^2}{ \pi \rme}  + 1 \right) \label{eq:choosingXuniform}\\
&= \frac12 \log \left(\frac{ 2 \min \left\{\sfA^2, 3 \sfP \right\}}{\pi \rme} + 1 \right) \text{,}
\end{align} 
where  \eqref{eq:choosingXuniform} follows by sub-optimally choosing $X$ to be uniform on $[-a,a]$. 
\qed

\section{Additional Lemmas for the Upper Bound Proof of Theorem~\ref{thm:AGC amplitude and power constraint}}\label{apdx:AGC amplitude-power supplement}
  Crucial to the proofs that follow, the next lemma provides a bound on the value of the Lagrange multiplier $\lambda$ in Smith's result \cite[Corollary 2]{smith1971information}.
\begin{lemma}[Bound on the Value of $\lambda$]\label{lem:lambda value}
The Lagrange multiplier $\lambda$ that appears in Lemma~\ref{lem:Smith's Lemma on Variance Constraint Problem} satisfies  
  \begin{align}
    \lambda &\le \frac{\log(1+ \sfP)}{2\sfP} \cdot 1\left\{\sfP < \sfA^2\right\}  \text{.}
  \end{align}
\end{lemma}

\begin{proof}
 If $\sfP \ge \sfA^2 $, the power constraint in \eqref{eq:channel power} is not active, implying that the Lagrange multiplier $\lambda = 0$. Suppose $\sfP < \sfA^2$. It follows from Lemma~\ref{lem:Smith's Lemma on Variance Constraint Problem} that
  \begin{align}
       \lambda  \sfP &\le C(\sfA, \sfP) - i(0; P_{X^\star})  \\ 
       &\le \frac12 \log(1+ \sfP) \text{,} \label{frm:cap of no const}
  \end{align}
where \eqref{frm:cap of no const} is because $ C(\sfA, \sfP) \le  C(\infty , \sfP) =\frac12 \log(1+ \sfP)$, and $i(0; P_{X^\star})  =  D(Z\|Y) \ge 0$.  
\end{proof}
  
Similar to its counterpart in Lemma~\ref{lem:LocationOFZeros_RealCase}, the next lemma provides a bound on the interval for the zeros of the function $\mfrF_{\sfA,\, \sfP}^\star $.
  
\begin{lemma}[Location and Finiteness of Zeros of $\rme^{\lambda y^2} f_Y(y) - \kappa_{\sfA,\, \sfP} $]\label{lem:location of zeros power constraint} 
For a fixed $\kappa_{\sfA,\, \sfP} \in \big(0, \frac{1}{\sqrt{2\pi}}\big] $, there exists some $\sfB_{\kappa_{\sfA,\, \sfP}} = \sfB_{\kappa_{\sfA,\, \sfP}}(\sfA, \sfP) < \infty $ such that 
\begin{align}
&  \rmN\left(\bbR, \rme^{\lambda y^2} f_Y(y) - \kappa_{\sfA,\, \sfP} \right) \notag\\
   &= \rmN\left([- \sfB_{\kappa_{\sfA,\, \sfP}}, \sfB_{\kappa_{\sfA,\, \sfP}}],\rme^{\lambda y^2} f_Y(y) - \kappa_{\sfA,\, \sfP} \right) \\
  &<\infty \text{.}
\end{align}
  In other words, there are finitely many zeros of $\rme^{\lambda y^2} f_Y(y) - \kappa_{\sfA,\, \sfP}$ which are contained within the interval $[-\sfB_{\kappa_{\sfA,\, \sfP}}, \sfB_{\kappa_{\sfA,\, \sfP}}]$. Moreover, 
\begin{align}
  \sfB_{\kappa_{\sfA,\, \sfP}} &\le \frac{\sfA}{1-2\lambda} + \left(\frac{1}{1-2\lambda} \log \frac{1}{2\pi\kappa_{\sfA,\, \sfP}^2} + \frac{2\lambda \sfA^2}{\left(1-2\lambda\right)^2} \right)^\frac12 \text{.}  \label{eq:upper bound on B_kappa-power const case} 
  \end{align}
\end{lemma}  
  
\begin{proof}
  Using the monotonicity of $\rme^{-u} $, for $|y| > \sfA $,
\begin{align}
   & \rme^{\lambda y^2} f_Y(y)  \notag\\
    &= \frac{\rme^{\lambda y^2} }{\sqrt{2\pi}} \bbE\left[\exp\left(-\frac{(y-X)^2}{2}\right) \right] \\
    &\le \frac{\rme^{\lambda y^2}}{\sqrt{2\pi}} \exp\left(-\frac{(y-\sfA)^2 }{2}\right) \\
    &= \frac{1}{\sqrt{2\pi}}  \exp\left(-\frac{1-2\lambda}{2} \left(y- \frac{\sfA}{1-2\lambda}\right)^2 + \frac{\lambda \sfA^2}{1-2\lambda}\right) \text{.} \label{frm:algebra for var case} 
  \end{align}
Since $\lambda \in [0, 1/2)$, cf. Lemma~\ref{lem:lambda value}, the right side of \eqref{frm:algebra for var case} is a decreasing function for all $|y| > \frac{A}{1-2\lambda}$ and we have
\begin{align}
   \rme^{\lambda y^2} f_Y(y) - \kappa_{\sfA,\, \sfP} < 0 
\end{align}
  for all 
\begin{align}
    |y|> \frac{\sfA}{1-2\lambda} + \left(\frac{2}{1-2\lambda} \log \frac{1}{\kappa_{\sfA,\, \sfP} \sqrt{2\pi}} + \frac{2\lambda \sfA^2}{\left(1-2\lambda\right)^2} \right)^\frac12 \text{.} \label{frm:equate to zero} 
  \end{align}
 This means that there exists $\sfB_{\kappa_{\sfA,\, \sfP}}$ satisfying \eqref{eq:upper bound on B_kappa-power const case} such that all zeros of $\mfrF_{\sfA,\, \sfP} (y) = \rme^{\lambda y^2} f_Y(y) - \kappa_{\sfA,\, \sfP}$ are contained within the interval $[-\sfB_{\kappa_{\sfA,\, \sfP}}, \sfB_{\kappa_{\sfA,\, \sfP}}]$. 
  
 To see the finiteness of the number of zeros of $\mfrF_{\sfA,\, \sfP} (y) $, it suffices to show that $\mfrF_{\sfA,\, \sfP}$ is analytic on $\bbR$ as analytic functions have finitely many zeros on a compact interval. However, it is easy to see that $\mfrF_{\sfA,\, \sfP}$ is analytic because convolution with a Gaussian preserves analyticity \cite[Proposition 8.10]{folland2013real}.
\end{proof}
\begin{lemma}\label{lem:looser bound on B_kappa}
For $ \kappa_{\sfA,\, \sfP}$ as defined in \eqref{eq:def:kappa_{A,P}}, the bound on the location of the zeros in Lemma~\ref{lem:location of zeros power constraint} can be loosened as 
  \begin{align}
    \sfB_{\kappa_{\sfA,\, \sfP}} < 2\sfA_\sfP + 1 \text{,}
  \end{align}
where 
 \begin{align}
   \sfA_\sfP = \frac{\sfA \sfP}{\sfP - \log(1+\sfP) \cdot 1\{\sfP < \sfA^2  \}} \text{.}
 \end{align} 
\end{lemma}
\begin{proof}
We may assume that $\sfP<\sfA^2 $, otherwise see \eqref{eqn:required later in the amp and pow const case}. In that case, observe that 
  \begin{align}
    C(\sfA, \sfP) \le \frac12 \log(1 + \sfP) \text{,}
  \end{align}
  and hence,
  \begin{align}
    \kappa_{\sfA,\, \sfP} &= \exp(-h(Z)-C(\sfA,\sfP)+\lambda(\sfP +1)) \\
    &\ge \frac{\exp(\lambda(\sfP +1))}{\sqrt{2\pi\rme (1+ \sfP)}} \text{.} \label{comb:bd on kappa AP}
  \end{align}
Combining \eqref{eq:upper bound on B_kappa-power const case} and \eqref{comb:bd on kappa AP} 
\begin{align}
    &\sfB_{\kappa_{\sfA,\, \sfP}}  \notag\\
    &\le  \frac{\sfA}{1-2\lambda} + \left(1+ \frac{\log(1+\sfP)}{1-2\lambda} + \frac{2\lambda}{1-2\lambda}\left(\frac{\sfA^2}{1-2\lambda} -\sfP\right)  \right)^\frac12  \label{bc:leftside} \\
    &\le \sfA_\sfP + (1+2\lambda\sfA_\sfP^2 )^{\frac12} \label{frm:lemma on val of lambd} \\
    &< 2\sfA_\sfP + 1 \text{,}  \label{frm:lambda <1/2}
\end{align}
where \eqref{frm:lemma on val of lambd} follows from Lemma~\ref{lem:lambda value} as the right side of \eqref{bc:leftside} is increasing in $\lambda \le \frac{\log(1+ \sfP)}{2\sfP} $; and \eqref{frm:lambda <1/2} follows because $\lambda< \frac12 $, cf. Lemma~\ref{lem:lambda value}. 
\end{proof}
\begin{lemma} \label{lem:upper bound power constraint}
  Suppose $\mfrF_{\sfA,\, \sfP} \colon \bbR \to \bbR $ is such that $\mfrF_{\sfA,\, \sfP} (y)  = \rme^{\lambda y^2} f_Y (y) - \kappa_{\sfA,\, \sfP} $. The complex extension of its derivative $\breve \mfrF_{\sfA,\, \sfP}' \colon \bbC \to \bbC $ satisfies
   \begin{align}
     &\max_{|z|\le \sfB} \left|\breve \mfrF_{\sfA,\, \sfP}' (z)\right|  \notag\\
     &\le \frac{1}{\sqrt{2\pi}} (\sfA+(1+2\lambda_\sfP)\sfB)\exp\left(\frac{(1+2\lambda_\sfP)\sfB^2}2\right)  \label{eqn:lem:lambdaP}  \\
     &< \frac{1}{\sqrt{2\pi}} (\sfA+2\sfB) \exp\left(\sfB^2\right)   \text{,} 
  \end{align}
  where in \eqref{eqn:lem:lambdaP} $\lambda_\sfP = \frac{\log(1+\sfP)}{2\sfP} \cdot 1\{\sfP < \sfA^2\}$.
\end{lemma}
\begin{proof}
  Denote by $\breve f_Y  $ and $\breve f_Y' $ the analytic complex extensions of the probability density function $f_Y $ and its derivative $f_Y' $, respectively. Then, 
    \begin{align}
     & \max_{|z|\le \sfB} \left|\breve \mfrF_{\sfA,\, \sfP}' (z)\right|  \notag\\
      &=  \max_{|z|\le \sfB} \left| \rme^{\lambda z^2} \left(\breve f'_Y(z) + 2 \lambda z \breve f_Y(z)\right) \right|  \label{frm:defs of mfrF and f_Y} \\
  &\le\rme^{\lambda \sfB^2 } \max_{|z|\le \sfB} \left|\breve f'_Y(z) + 2 \lambda z\breve f_Y(z)\right|  \label{frm:|z|<sfB} \\
   &\le\rme^{\lambda \sfB^2 } \left(\max_{|z|\le \sfB} \left|\breve f'_Y(z)\right| + \max_{|z|\le \sfB}\left|2 \lambda z\breve f_Y(z)\right|\right) \label{frm:triangle ineq and max max} \\
   &\le\frac{\rme^{\lambda \sfB^2 }}{\sqrt{2\pi}} \left((\sfA+\sfB)\exp\left(\frac{\sfB^2}2\right) \right. \\
   & \left. \quad + \max_{|z|\le \sfB}\left|2\lambda z \bbE\left[\exp\left(-\frac{(z-X)^2}{2}\right)\right]\right|  \right) \label{frm:prev lemma in amp case} \\
   &\le\frac{\rme^{\lambda \sfB^2 }}{\sqrt{2\pi}} \left((\sfA+\sfB)\exp\left(\frac{\sfB^2}2\right)  \right. \\
   & \left. \quad + 2\lambda \sfB \max_{|z|\le \sfB}\bbE\left[\left|\exp\left(-\frac{(z-X)^2}{2}\right)\right|\right]  \right) \label{frm:alx says modulus ineq} \\
   &\le\frac{\rme^{\lambda \sfB^2 }}{\sqrt{2\pi}} \left((\sfA+\sfB)\exp\left(\frac{\sfB^2}2\right) + 2\lambda \sfB \max_{|z|\le \sfB} \exp\left(\frac{\mathfrak{Im}^2(z)}{2}\right) \right) \label{frm:prv/ampl case alne} \\  
   &\le \frac{1}{\sqrt{2\pi}} (\sfA+(1+2\lambda)\sfB)\exp\left(\frac{(1+2\lambda)\sfB^2}2\right)   \text{,}
\end{align}
 where \eqref{frm:defs of mfrF and f_Y} follows from definitions of the functions involved; \eqref{frm:|z|<sfB} is because $|z| \le \sfB $ implies $\big|\rme^{\lambda z^2 }\big| \le \rme^{\lambda \sfB^2 } $; \eqref{frm:triangle ineq and max max} follows from the triangle inequality; \eqref{frm:prev lemma in amp case} follows from Lemma~\ref{lem:upper bound on maxima over a disk}; \eqref{frm:alx says modulus ineq} follows from the modulus inequality.
 
 The desired result is a consequence of the fact that the Lagrange multiplier satisfies $\lambda\le \lambda_\sfP < 1/2 $, cf. Lemma~\ref{lem:lambda value}. 
\end{proof}
 
 \begin{lemma} \label{lem:lower bound power constraint}
Let $\sfB \ge \sfA/(1-2\lambda) $.  Suppose $\mfrF_{\sfA,\, \sfP} \colon \bbR \to \bbR $ is such that $\mfrF_{\sfA,\, \sfP} (y)  = \rme^{\lambda y^2} f_Y (y) - \kappa_{\sfA,\, \sfP} $. The complex extension of its derivative $\breve \mfrF_{\sfA,\, \sfP}' \colon \bbC \to \bbC $ satisfies
  \begin{align}
    \max_{|z|\le \sfB} \left|\breve \mfrF_{\sfA,\, \sfP}'(z) \right| &\ge \frac{\sfA}{\sqrt{2\pi}} \exp\left(- \frac{2- \lambda_\sfP}{1-2\lambda_\sfP}\, \sfA^2\right)     \text{,}
  \end{align}
  where $\lambda_\sfP = \frac{\log(1+\sfP)}{2\sfP} \cdot 1\{\sfP < \sfA^2 \} $. 
\end{lemma} 

\begin{figure*}[h!]
  \begin{align}
     &\rmN\left([-\sfB_{\kappa_{\sfA,\, \sfP}}, \sfB_{\kappa_{\sfA,\, \sfP}}], \mfrF_{\sfA,\, \sfP}\right) \nonumber \\
     &\le 1+ \rmN\left([-\sfB_{\kappa_{\sfA,\, \sfP}}, \sfB_{\kappa_{\sfA,\, \sfP}}], \mfrF_{\sfA,\, \sfP}'\right)  \label{frm:Rolle again} \\
     &\le 1+ \rmN\left(\cD_{\sfB_{\kappa_{\sfA,\, \sfP}}},\breve \mfrF_{\sfA,\, \sfP}'\right) \label{frm:complex mfrF has more zeros} \\
     &\le 1+ \min_{s>1,\, t\ge \frac{\sfA_\sfP}{\sfB_{\kappa_{\sfA,\, \sfP}}} } \left\{\frac{1}{\log s} \left( \log \max_{|z|\le (st+s+t)\sfB_{\kappa_{\sfA,\, \sfP}}}\left|\breve \mfrF_{\sfA,\, \sfP}'\right| - \log \max_{|z|\le t \sfB_{\kappa_{\sfA,\, \sfP}}} \left|\breve \mfrF_{\sfA,\, \sfP}'\right|  \right) \right\} \label{frm:Tijdeman nu of zero} \\ 
     &\le 1+ \min_{s>1,\, t\ge \frac{\sfA_\sfP}{\sfB_{\kappa_{\sfA,\, \sfP}}}} \left\{\frac{1}{\log s} \left( \frac{(st+s+t)^2\sfB_{\kappa_{\sfA,\, \sfP}}^2}{2/(1+2\lambda_\sfP) } +  \frac{2- \lambda_\sfP}{1-2\lambda_\sfP}\, \sfA^2 + \log\left(1+\frac{(st+s+t)\sfB_{\kappa_{\sfA,\, \sfP}}}{\sfA/(1+2\lambda_\sfP) }\right) \right)  \right\} \label{frm:invoke lem apd upp and low} \\
     &= 1+ \min_{s>1} \left\{\frac{1}{\log s} \left( \frac{((\sfA_\sfP +\sfB_{\kappa_{\sfA,\, \sfP}})s + \sfA_\sfP)^2 }{2/(1+2\lambda_\sfP) } +  \frac{2- \lambda_\sfP}{1-2\lambda_\sfP}\, \sfA^2 + \log\left(\frac{2}{1-2\lambda_\sfP} +\frac{(\sfA_\sfP +\sfB_{\kappa_{\sfA,\, \sfP}})s }{\sfA/(1+2\lambda_\sfP) }\right) \right)  \right\} \label{frm:optmizing t is the smallest} \\
     &\le 1+ \min_{s>1} \left\{\frac{1}{\log s} \left( \frac{((3\sfA_\sfP +1)s + \sfA_\sfP)^2 }{2/(1+2\lambda_\sfP) } +  \frac{2- \lambda_\sfP}{1-2\lambda_\sfP}\, \sfA^2 + \log\left(\frac{2}{1-2\lambda_\sfP} +\frac{(3\sfA_\sfP +1)s }{\sfA/(1+2\lambda_\sfP) }\right) \right)  \right\} \label{frm:looser bound on B_kappa} \\
     &\le 1+ 2 \left( \frac{((3\sqrt{\rme} +1)\sfA_\sfP + \sqrt{\rme} )^2 }{2/(1+2\lambda_\sfP) } +  \frac{2- \lambda_\sfP}{1-2\lambda_\sfP}\, \sfA^2 + \log\left(\frac{2}{1-2\lambda_\sfP} +\frac{(3\sfA_\sfP +1) \sqrt{\rme} }{\sfA/(1+2\lambda_\sfP) }\right) \right) \label{frm:sub opt sqrt{e}} \\
     &\le 1+ 2 \left( \frac{((3\sqrt{\rme} +1)\sfA_\sfP + \sqrt{\rme} )^2 }{2/(1+2\lambda_\sfP) } + (2- \lambda_\sfP)(1-2\lambda_\sfP)\, \sfA_\sfP^2 + \log\left(\frac{2+4 \sqrt{\rme}(1+2\lambda_\sfP)}{1-2\lambda_\sfP} \right) \right) \label{frm:assuming A>1} \text{,}
  \end{align} 
\rule{\textwidth}{0.5pt}
  \end{figure*}

\begin{proof}
  Note that
  \begin{align}
 &    \notag\\
    &\max_{|z|\le \sfB} \left|\breve \mfrF_{\sfA,\, \sfP}'(z) \right| \\ 
    &= \max_{|z|\le \sfB} \left|\rme^{\lambda z^2} \left(\breve f'_Y(z) + 2 \lambda z \breve f_Y(z)\right)\right| \\
     &=\max_{|z|\le \sfB} \frac{1}{\sqrt{2\pi}}\left|\bbE\left[(X-(1-2\lambda)z)\rme^{-\frac{1-2\lambda}{2} \left(z- \frac{X}{1-2\lambda}\right)^2 + \frac{\lambda X^2}{1-2\lambda}}\right]\right| \\ 
     &\ge \frac{1}{\sqrt{2\pi}} \left|\bbE\left[(X-\sfA)\exp\left( \frac{-\sfA^2+2\sfA X - (1-2\lambda)X^2}{2-4\lambda} \right)\right] \right| \label{frm:subopt z=A/(1-lambda)}  \\  
     &\ge \frac{1}{\sqrt{2\pi}} \bbE\left[(\sfA-X)\exp\left( \frac{-3\sfA^2 }{2-4\lambda} - \frac12 \sfA^2 \right)\right] \label{frm:|X|<A again}\\
      &= \frac{\sfA}{\sqrt{2\pi}} \exp\left(-\left(\frac12 + \frac{3}{2-4\lambda}\right)\sfA^2  \right) \label{bc:E[X]=0} \\
      &\ge \frac{\sfA}{\sqrt{2\pi}} \exp\left(- \frac{2- \lambda_\sfP}{1-2\lambda_\sfP}\, \sfA^2\right) \label{bc:lamda<log(1+p)/2p} \text{,}
  \end{align}
  where \eqref{frm:subopt z=A/(1-lambda)} follows from the suboptimal choice of $z = \frac{\sfA}{1-2\lambda} \le \sfB $; \eqref{frm:|X|<A again} follows because $|X|\le \sfA $; \eqref{bc:E[X]=0} is a consequence of $\bbE[X]=0$; and finally, \eqref{bc:lamda<log(1+p)/2p} follows because $\lambda \le \lambda_\sfP $, see Lemma~\ref{lem:lambda value}. 
  \end{proof}

\begin{lemma}[Bound on the Number of Oscillations of $\mfrF_{\sfA,\, \sfP}^\star$] \label{lem:bound using tijdeman power case} 
Suppose that  $\mfrF_{\sfA,\, \sfP}^\star$ is as defined in \eqref{eq:de:mfrF_{A,P}} and $\sfB_{\kappa_{\sfA,\, \sfP}}$ be as defined in Lemma~\ref{lem:location of zeros power constraint}. The number of zeros of $\mfrF_{\sfA,\, \sfP}^\star $ within the interval $[-\sfB_{\kappa_{\sfA,\, \sfP}}, \sfB_{\kappa_{\sfA,\, \sfP}}]$ satisfies
\begin{align}
  \rmN\left([-\sfB_{\kappa_{\sfA,\, \sfP}}, \sfB_{\kappa_{\sfA,\, \sfP}}], \mfrF_{\sfA,\, \sfP}\right) \le  a^{}_{\sfP_2} \sfA_\sfP^2 + a^{}_{\sfP_1} \sfA_\sfP + a^{}_{\sfP_0}  \text{,}\label{eq:lem:upp bd on support}
\end{align}
where $\sfA_\sfP $, $ a^{}_{\sfP_2}$, $ a^{}_{\sfP_1}$, $ a^{}_{\sfP_0}$, and $\lambda_\sfP$ are as defined in \eqref{eqn:def:A_P}, \eqref{eqn:def:a_{sfP_2}}, \eqref{eqn:def:a_{sfP_1}}, \eqref{eqn:def:a_{sfP_0}}, and \eqref{eqn:def:lambda_P}. 
\end{lemma}
\begin{proof}
We may assume $\sfP < \sfA^2$, otherwise see the proof of the upper bound in Theorem~\ref{thm:MainResult}. For an arbitrary output density $f_Y$ define $\mfrF_{\sfA,\, \sfP} \colon \bbR \to \bbR $ such that $\mfrF_{\sfA,\, \sfP} (y)  = \rme^{\lambda y^2} f_Y (y) - \kappa_{\sfA,\, \sfP}$ and let $\mfrF_{\sfA,\, \sfP}'\colon \bbR \to \bbR $ be the derivative of $\mfrF_{\sfA,\, \sfP} $. Consider the disk $\cD_R \subset \bbC $ of radius $R$ centered at the origin and note 
 the following sequence of inequalities shown at the top of this page.

There, \eqref{frm:Rolle again} follows from Rolle's Theorem, see Lemma~\ref{lem:extreme points and oscillations}; \eqref{frm:complex mfrF has more zeros} follows because the zeros of $\mfrF_{\sfA,\, \sfP}' \colon \bbR \to \bbR $ are also the zeros of its complex extension $\breve \mfrF_{\sfA,\, \sfP}' \colon \bbC \to \bbC $; \eqref{frm:Tijdeman nu of zero} is a consequence of Tijdeman's Number of Zeros Lemma, namely Lemma~\ref{lem:number of zeros of analytic function}; \eqref{frm:invoke lem apd upp and low} follows by invoking Lemmas~\ref{lem:upper bound power constraint}~and~\ref{lem:lower bound power constraint} above; \eqref{frm:optmizing t is the smallest} follows because $t = \frac{\sfA_\sfP}{\sfB_{\kappa_{\sfA,\, \sfP}}} $ is the minimizer in the right side of \eqref{frm:invoke lem apd upp and low}; \eqref{frm:looser bound on B_kappa} follows from the fact that $\sfB_{\kappa_{\sfA,\, \sfP}}< 2\sfA_\sfP +1 $, see Lemma~\ref{lem:looser bound on B_kappa}; \eqref{frm:sub opt sqrt{e}} is a consequence of the suboptimal choice $s = \sqrt{\rme} $; and finally, \eqref{frm:assuming A>1} follows from the assumption that $\sfA > 1 $. 
 
 Algebraic manipulations in the right side of \eqref{frm:assuming A>1} yield the desired result in \eqref{eq:lem:upp bd on support}. 
\end{proof}

\end{appendices}

\bibliographystyle{IEEEtran}
\bibliography{refs_small.bib}

\begin{IEEEbiographynophoto}{Alex Dytso}  is currently a Postdoctoral Researcher in the Department of Electrical Engineering at Princeton University. In 2016, he received a Ph.D. degree from the Department of Electrical and Computer Engineering at the University of Illinois, Chicago. He received his B.S. degree in 2011 from the University of Illinois, Chicago, where he also received the International Engineering Consortium's William L. Everitt Student Award of Excellence for outstanding seniors. His current research interest are in the areas of multi-user information theory and estimation theory, and their applications in wireless networks.\end{IEEEbiographynophoto}

 \smallskip{}
\begin{IEEEbiographynophoto}{Semih Yagli}
 received his Bachelor of Science degree in Electrical and Electronics Engineering in 2013, his Bachelor of Science degree in Mathematics in 2014 both from Middle East Technical University and his Master of Arts degree in Electrical Engineering in 2016 from Princeton University.
 Currently, he is pursuing his Ph.D. degree in Electrical Engineering at Princeton University under the supervision of H. Vincent Poor. His research interest include information theory, optimization, statistical modeling, privacy and unsupervised machine learning. 
\end{IEEEbiographynophoto}

 \smallskip{}
\begin{IEEEbiographynophoto}{H. Vincent Poor} (S'72, M'77, SM'82, F'87) received the Ph.D. degree in electrical engineering and computer science from Princeton University in 1977.  From 1977 until 1990, he was on the faculty of the University of Illinois at Urbana-Champaign. Since 1990 he has been on the faculty at Princeton, where he is the Michael Henry Strater University Professor of Electrical Engineering. During 2006 to 2016, he served as Dean of Princeton's School of Engineering and Applied Science. He has also held visiting appointments at several other institutions, most recently at Berkeley and Cambridge. His research interests are in the areas of information theory and signal processing, and their applications in wireless networks, energy systems and related fields. Among his publications in these areas is the recent book {\it Information Theoretic Security and Privacy of Information Systems} (Cambridge University Press, 2017).

Dr. Poor is a member of the National Academy of Engineering and the National Academy of Sciences, and is a foreign member of the Chinese Academy of Sciences, the Royal Society, and other national and international academies. Recent recognition of his work includes the 2017 IEEE Alexander Graham Bell Medal, the 2019 ASEE Benjamin Garver Lamme Award, a D.Sc. {\it honoris causa} from Syracuse University, awarded in 2017, and a D.Eng. {\it honoris causa} from the University of Waterloo, awarded in 2019.

\end{IEEEbiographynophoto}
 \smallskip{}

\begin{IEEEbiographynophoto}{Shlomo Shamai (Shitz)}
Shlomo Shamai (Shitz) is with the Department of Electrical Engineering, 
Technion---Israel Institute of Technology, where he is now a 
Technion Distinguished Professor, and holds the William Fondiller 
Chair of Telecommunications.
His research interests encompasses a wide spectrum of topics in information
theory and statistical communications.
Dr. Shamai (Shitz) is an IEEE Life Fellow, an URSI Fellow, a member of the 
Israeli Academy of Sciences and Humanities and a foreign member of the
US National Academy of Engineering. He is the recipient of the 2011
Claude E. Shannon Award, the 2014 Rothschild Prize in
Mathematics/Computer Sciences and Engineering and the
2017 IEEE Richard W. Hamming Medal. He is a co-recipient of the 2018 Third
Bell Labs Prize for Shaping the Future of Information and Communications 
Technology.

He is the recipient of numerous technical and paper awards
and recognitions and is listed as a Highly Cited Researcher (Computer Science)
for the years 2013/4/5/6/7/8.
He has served as Associate Editor for the Shannon Theory of the IEEE
Transactions on Information Theory, and has also served twice on the
Board of Governors of the Information Theory Society.
He has also served on the Executive Editorial Board of the IEEE Transactions
on Information Theory, the IEEE Information Theory Society Nominations
and Appointments Committee and the IEEE Information Theory Society,
Shannon Award Committee.
\end{IEEEbiographynophoto}

\end{document}